\documentclass[11pt,a4paper]{article}
\usepackage{amsthm}
\usepackage{amsmath}
\usepackage{amssymb}
\usepackage{microtype}
\usepackage{bbm}
\usepackage[all]{xy}
\usepackage{bbm}
\usepackage{verbatim}
\usepackage{mathrsfs}
\usepackage[margin=1in]{geometry}
\usepackage{algorithm}
\usepackage{enumitem}
\usepackage{todonotes}
\usepackage{color}
\usepackage{hyperref}
\hypersetup{colorlinks=true,citecolor=blue,linkcolor=blue,urlcolor=blue}
\usepackage{tikz}
\usetikzlibrary{patterns,matrix,calc}
\definecolor{Csmug}{RGB}{170,220,150}
\tikzset{vsmug/.style={pattern=north east lines,pattern color=Csmug,rounded corners=1}}

\newtheorem{thm}{Theorem}[section]
\newtheorem*{thm*}{Theorem}
\newtheorem{lem}[thm]{Lemma}
\newtheorem{cor}[thm]{Corollary}
\newtheorem{prop}[thm]{Proposition}
\newtheorem*{prop*}{Proposition}

\theoremstyle{definition}
\newtheorem{defi}[thm]{Definition}

\newtheorem{rmk}[thm]{Remark}
\usepackage[noend]{algpseudocode}

\usepackage{filecontents}

\DeclareMathOperator{\sel}{sel}
\DeclareMathOperator{\Pol}{Pol}
\DeclareMathOperator{\PMC}{PMC}
\DeclareMathOperator{\arOperator}{ar}
\renewcommand{\ar}{\arOperator}
\newcommand{\Mm}{\mathcal{M}}
\DeclareMathOperator{\PCSP}{PCSP}

\newcommand{\A}{\mathbf{A}}
\newcommand{\B}{\mathbf{B}}
\newcommand{\C}{\mathbf{C}}

\usepackage[backend=bibtex,bibencoding=ascii,style=alphabetic,bibstyle=alphabetic,firstinits=true,doi=false,isbn=false,url=false,maxbibnames=99]{biblatex}
\renewbibmacro{in:}{}
\AtEveryBibitem{\clearname{editor}}

\newbibmacro{string+link}[1]{%
	\iffieldundef{ee}
	{\iffieldundef{url}
		{\iffieldundef{doi}
			{\iffieldundef{eprint}
				{#1}
				{\href{http://arxiv.org/abs/\thefield{eprint}}{#1}}}
			{\href{http://dx.doi.org/\thefield{doi}}{#1}}}
		{\href{\thefield{url}}{#1}}}
	{\href{\thefield{ee}}{#1}}} 
\DeclareFieldFormat{title}{\usebibmacro{string+link}{\mkbibemph{#1}}}
\DeclareFieldFormat[article,inproceedings,inbook,incollection,mastersthesis,thesis]{title}{\usebibmacro{string+link}{\mkbibquote{#1}}}
\DeclareFieldFormat{eprint:eccc}{ECCC\addcolon\space\ifhyperref
    {\href{http://eccc.hpi-web.de/report/#1/}{\nolinkurl{#1}}}
    {\nolinkurl{#1}}}

\begin{document}

\author{Alex Brandts\\
University of Oxford\\
\texttt{alex.brandts@cs.ox.ac.uk}
\and
Marcin Wrochna\\
University of Oxford\\
\texttt{marcin.wrochna@cs.ox.ac.uk}
\and
Stanislav \v{Z}ivn\'y\\
University of Oxford\\
\texttt{standa.zivny@cs.ox.ac.uk}
}

\title{The complexity of promise SAT on non-Boolean domains\thanks{
An extended abstract of part of this work appeared in the Proceedings of
ICALP'20~\cite{BWZ20}.
Alex Brandts was supported by a Royal Society Enhancement Award and an NSERC PGS Doctoral Award. Stanislav \v{Z}ivn\'y was supported by a Royal Society University
Research Fellowship. This project has received funding from the European
Research Council (ERC) under the European Union's Horizon 2020 research and
innovation programme (grant agreement No 714532). The paper reflects only the
authors' views and not the views of the ERC or the European Commission. The
European Union is not liable for any use that may be made of the information
contained therein.}}

\maketitle

\begin{abstract}
  
While 3-SAT is NP-hard, 2-SAT is solvable in polynomial time.
Austrin, Guruswami, and H{\aa}stad [FOCS'14/SICOMP'17] proved a result known as ``$(2+\varepsilon)$-SAT is NP-hard''.
They showed that the problem of distinguishing $k$-CNF formulas that are $g$-satisfiable (i.e. some assignment satisfies at least $g$ literals in every clause)
from those that are not even 1-satisfiable is NP-hard if $\frac{g}{k} < \frac{1}{2}$ and is in P otherwise.
We study a generalisation of SAT on arbitrary finite domains,
with clauses that are disjunctions of unary constraints,
and establish analogous behaviour.
Thus we give a dichotomy for a natural fragment of
promise constraint satisfaction problems (PCSPs) on arbitrary finite domains.

The hardness side is proved using the algebraic approach via a new general
NP-hardness criterion on polymorphisms, which is based on a gap version
of the Layered Label Cover problem. We show that previously used criteria are
insufficient -- the problem hence gives an interesting benchmark of algebraic
techniques for proving hardness of approximation in problems such as PCSPs.
\end{abstract}

\section{Introduction}
\label{sec:intro}

It is a classic result that while 3-SAT is NP-hard~\cite{c71,l73}, 2-SAT can be
solved in polynomial-time~\cite{k67}. Austrin, Guruswami, and H{\aa}stad~\cite{agh17}
considered the promise problem $(1,g,k)$-SAT (for integers $1\leq g \leq k$):
given a $k$-CNF formula with the promise that
there is an assignment that satisfies at least $g$ literals in each clause,
find an assignment that satisfies at least one literal in each clause.
They showed that the problem is NP-hard if $\frac{g}{k} < \frac{1}{2}$
and in P otherwise.
Viewing $k$-SAT as $(1,1,k)$-SAT, this shows that, in a natural sense,
the transition from tractability to hardness
occurs just after 2 and not just before 3.

The \emph{set-satisfiability} (SetSAT) problem generalises the Boolean satisfiability problem to larger domains and we prove that it exhibits an analogous hardness transition.
As in $(a,g,k)$-SAT, for integer constants $1 \leq a \leq g \leq k$ and $1 \leq s < d$,
an instance of the $(a,g,k)$-SetSAT problem is a conjunction of clauses,
where each clause is a disjunction of $k$ literals.
However, variables $x_1,\dots,x_n$ can take values in a larger domain $[d]=\{1,\dots,d\}$,
while literals take the form ``$x_i \in S$'', where $S$ is any subset of the domain $[d]$ of size $s$.
As in the Boolean case, an assignment $\sigma \colon \{x_1,\dots,x_n\} \to [d]$ is $g$-satisfying if it satisfies at least $g$ literals in every clause.
In $(a,g,k)$-SetSAT with set size $s$ and domain size $d$,
given an instance promised to be $g$-satisfiable,
we are to find an $a$-satisfying assignment.
When $s=1$ and $d=2$ we recover Boolean promise SAT,
whereas when $a=g=1$ we recover the non-promise version of SetSAT.

The most natural case of SetSAT is when we allow \emph{all} nontrivial unary constraints (sets) as literals,
i.e., the case $s=d-1$.
(While we defined sets specifying literals to have size exactly $s$,
one can simulate sets of size at most $s$ by replacing them with all possible supersets of size $s$; see the proof of Proposition~\ref{setsize}.)
More generally one could consider the problem restricted to any family of literals.
Our work deals with the ``folded'' case: if a set $S$ is available as a literal, then for all permutations of the domain $\pi$, $\pi(S)$ is also available as a literal.
In this case only the cardinality of $S$ matters, and in fact only the maximum available cardinality matters,
so all such problems are equivalent to $(a,g,k)$-SetSAT, for some constants $a,g,k,s,d$.

\subsection{Related work} 

Our main motivation to study SetSAT as a promise problem is the fact that it constitutes a
natural fragment of so-called promise constraint satisfaction problems (PCSPs),
which are problems defined by homomorphisms between relational structures (see
Section~\ref{sec:prelims} for more details).
PCSPs were studied as early as in the classic work of Garey and
Johnson~\cite{Garey76:jacm} on approximate graph colouring, but a systematic
study originated in the paper of Austrin et al.~\cite{agh17}.
In a series of papers~\cite{bg16:ccc,bg18,bg19}, Brakensiek
and Guruswami linked PCSPs to the universal-algebraic methods developed for the
study of non-uniform CSPs~\cite{bkw17:survey}. In particular, the notion of (weak)
polymorphisms, formulated in~\cite{agh17}, allowed some ideas 
developed for CSPs to be used in the context of PCSPs. The algebraic
theory of PCSPs was then lifted to an abstract level by Barto, Bul\'in, Krokhin, and
Opr\v{s}al in~\cite{bko19, bbko19}. Consequently, this theory was used by Ficak,
Kozik, Ol\v{s}\'ak, and Stankiewicz to obtain a dichotomy for symmetric Boolean
PCSPs~\cite{fkos19}, thus improving on an earlier result
from~\cite{bg18}, which gave a dichotomy for symmetric Boolean PCSPs
with folding. Further recent results on PCSPs
include the work of Krokhin and Opr\v{s}al~\cite{ko19:focs}, Brakensiek and
Guruswami~\cite{bg20}, and Austrin, Bhangale, and Potukuchi~\cite{abp20}.

Variants of the Boolean satisfiability problem over larger domains have been
defined using CNFs by Gil, Hermann, Salzer, and Zanuttini~\cite{ghsz08} and DNFs
by Chen and Grohe~\cite{cg10} but, as far as we are
aware, have not been studied as promise problems before.

\subsection{Results}

We completely resolve the complexity of $(a,g,k)$-SetSAT.

\begin{thm}\label{thm:general}
	Let $1\leq s < d$ and $1\leq g \leq k$.
	The problem $(a,g,k)$-SetSAT with set size $s$ and domain size $d$ is solvable
  in polynomial time if $\frac{g-a+1}{k-a+1} \geq \frac{s}{s+1}$ and is
  NP-complete otherwise. 
\end{thm}

Theorem~\ref{thm:general} easily follows, as described in
Section~\ref{sec:general}, from our main result, in which we show that the complexity of
$(1,g,k)$-SetSAT depends only on the ratio $\frac{g}{k}$. 

\begin{thm}\label{thm:main}
$(1,g,k)$-SetSAT with set size $s$ and domain size $s+1$ is solvable in polynomial time if
  $\frac{g}{k}\geq\frac{s}{s+1}$ and is NP-complete otherwise.
\end{thm}

Theorem~\ref{thm:main} generalises the case of $(1,g,k)$-SAT studied in~\cite{agh17}, where
$s=1$ and the hardness threshold is $\frac{1}{2}$. 

The positive side of Theorem~\ref{thm:main} is proved in
Section~\ref{sec:tractability} by a simple randomised algorithm
based on classical work of Papadimitriou~\cite{pap1991}, just as in the Boolean
case~\cite{agh17}. We also establish the (non-)applicability
of certain convex relaxations for SetSAT.

The main difficulty is in proving NP-hardness when the ratio $\frac{g}{k}$ is
close to, but below $\frac{s}{s+1}$. In Appendix~\ref{app:simple}, we show that
for certain (but not all) choices of the parameters NP-hardness can be derived
easily by simple gadget constructions and via a result of Guruswami and Lee on
hypergraph colourings~\cite{gl18}.

Following~\cite{agh17} and the more abstract algebraic framework of~\cite{bbko19},
the hardness proof relies on understanding polymorphisms, i.e.,
high-arity functions $f \colon [d]^n \to [d]$ which describe
the symmetries of our computational problem.
In the Boolean case, the proof of \cite{agh17} relies on showing that every polymorphism
depends on only a few variables (in other words, is a junta),
and that this condition suffices for a reduction from the Gap Label Cover problem.
In our case, this condition does not hold, and neither do the various generalisations of it used in later work on PCSPs~\cite{fkos19,bbko19,ko19:focs}.
In fact, we show in Section~\ref{sec:impos} that SetSAT has significantly richer, more robust polymorphisms,
which makes the application of many such conditions impossible.
Our main technical contribution is a new condition that guarantees an NP-hardness
reduction from a multilayered variant of the Gap Label Cover problem.

As in previous work, the combinatorial core of our NP-hardness results for SetSAT
relies on identifying, in every polymorphism $f \colon [d]^n \to [d]$,
a small set of distinguished coordinates.
The rough idea is that a polymorphism encodes a $1$-in-$n$ choice analogously to the long code,
and the reduction relies on being able to decode $f$ with small ambiguity.

The set of distinguished coordinates could be, in the simplest case, those on which $f$ depends
(called \emph{essential coordinates}) and, as shown in \cite{agh17}, a small set
of essential coordinates is sufficient for hardness of $(1,g,k)$-SAT if
$\frac{g}{k}<\frac{1}{2}$.
More generally, the distinguished set $S$ could be such that some partial assignment 
to $S$ makes $f$ constant (as a function of its remaining coordinates), or restricts the range of $f$ (called \emph{fixing}~\cite{agh17,fkos19} and \emph{avoiding}~\cite{bbko19} sets, respectively).
As shown in Section~\ref{sec:impos}, the polymorphisms of SetSAT on non-Boolean
domains do not have small sets of coordinates that are essential, fixing, or
avoiding. Instead, in this paper we introduce the notion of a \emph{smug set} of $f$.
We say that a set $S \subseteq [n]$ is smug if 
for some input $(a_1,\ldots,a_n)$ to $f$,
the coordinates $i$ whose values $a_i$
agree with the output $f(a_1,\ldots,a_n)$ are exactly those in $S$.
We show that every polymorphism of SetSAT has a smug set of constant size (independent of $n$)
and cannot have many disjoint smug sets.

In previous work, it was crucial that essential coordinates respect minors.
We say that (an $m$-ary function) $g$ is a \emph{minor} of (an $n$-ary function) $f$ if $g(x_1,\dots,x_m) \approx f(x_{\pi(1)},\dots,x_{\pi(n)})$ for some $\pi \colon [n] \to [m]$ (that is, $g$ is obtained from $f$ by identifying or permuting coordinates of $f$, or introducing inessential coordinates).
In that case, if $S$ contains all essential coordinates of $f$, then $\pi(S)$ contains all essential coordinates of $g$.
This does not hold for smug sets; instead, if $S$ is a smug set of $g$, then its pre-image $\pi^{-1}(S)$ is a smug set of $f$.
The pre-image may however be much larger.
Still, these properties of smug sets are enough to guarantee that, in any sufficiently long chain of minors, if one chooses a random coordinate in a small smug set from each function in the chain,
then for some two functions in the chain the choices will agree, respecting the minor relation between them
with constant probability.
We show that this condition is sufficient to obtain NP-hardness from a Gap Layered Label Cover problem.
See Section~\ref{sec:labelcover} for details.

We note that several other properties of Label Cover variants were used before in the context of polymorphisms.
Guruswami and Sandeep~\cite{GS20:rainbow} use ``smoothness'' of NP-hard Label Cover instances
(introduced by Khot~\cite{Khot02})
so that a minor relation~$\pi$ needs to be respected only if it is injective on a small set $S$.
This allows them to use sets~$S$ which are \emph{weakly fixing}, i.e. the partial assignment to $S$ which makes $f$ constant does not necessarily have to assign the same value to all coordinates in $S$.
Layered Label Cover was introduced by Dinur, Guruswami, Khot, and Regev~\cite{DinurGKR05} to tighten the approximation hardness for hypergraph vertex cover.
In the proof of hardness of hypergraph colouring by Dinur, Regev, and
Smyth~\cite{Dinur05:combinatorica}, as reinterpreted in~\cite{bbko19}, Layered Label Cover
is used to partition polymorphisms into an arbitrary constant number $L$ of parts, so that only minors within one part need to be respected.
This implies that in any chain of minors with $L+1$ functions, some two functions will be in the same part and hence the minor between them will be respected; our approach is hence similar, though apparently more general, in this aspect.
Another feature used in~\cite{Dinur05:combinatorica,bbko19} is that the bound on the size of a set of special coordinates or on the number of disjoint such sets may be any subpolynomial function in $n$, not necessarily a constant.
These different features of NP-hard Label Cover instances can be combined; however, this is not necessary for our result.

\section{Preliminaries}
\label{sec:prelims}

Let $[n] = \{1,2,\ldots,n\}$. 
For a set $A$, we call $R \subseteq A^k$ a \emph{relation} of arity $\ar(R) =k$
and $f \colon A^k \to B$ a function of arity $\ar(f)=k$. 

We take the domain of the variables in SetSAT to be $[d]$ and
for a fixed $s < d$ we identify each literal with the indicator function of
some $S \subseteq [d], |S|=s$: $S(x)= \mathbbm{1}[x \in S]$. For a SetSAT
instance (or \emph{formula}) with $n$ variables $x_1,\ldots,x_n$, an assignment
to the variables is a function $\sigma : \{x_1,\ldots,x_n\} \to [d]$. An assignment $\sigma$ is called a $g$-satisfying assignment for an instance $\Psi$ if $\sigma$ satisfies at least $g$ literals in every clause of $\Psi$. A 1-satisfying assignment is usually simply called a satisfying assignment. A formula is called $g$-satisfiable if there exists a $g$-satisfying assignment to its variables, and satisfiable if there exists a 1-satisfying assignment. 

The SAT problem corresponds to the SetSAT problem with $d=2$ and $s=1$, so SetSAT does indeed generalise SAT. Note that every SetSAT instance is trivially unsatisfiable when $s=0$ and satisfiable when $s=d$, so we exclude these cases in our analysis. We now give the formal definition of $(a,g,k)$-SetSAT.

\begin{defi}\label{SetSATdef}
Let $1 \leq s < d$ and $1 \leq a \leq g \leq k$. The $(a,g,k)$-SetSAT problem is
  the following promise problem. In the decision version, given a SetSAT
  instance where each clause has $k$ (not necessarily distinct) literals, accept the instance if it is
  $g$-satisfiable and reject it if it is not $a$-satisfiable. In the search
  version, given a $g$-satisfiable SetSAT instance, find an $a$-satisfying assignment.
\end{defi}

We will prove hardness only for the decision version of $(a,g,k)$-SetSAT and
tractability only for the search version. This suffices since the decision
version of $(a,g,k)$-SetSAT is polynomial-time reducible to the corresponding
search problem. This is discussed in the context of PCSPs
in~\cite{bko19,bbko19}. For completeness, we give the reduction. Suppose we are
given a SetSAT formula $\Psi$. We run the search algorithm on $\Psi$, and check
that the output of the algorithm does indeed $a$-satisfy $\Psi$. If it does,
accept $\Psi$; otherwise, reject it. Since in the decision problem we are
guaranteed that the input is either $g$-satisfiable or not even $a$-satisfiable,
the algorithm is correct in both cases. Therefore, the algorithmic result in
Proposition~\ref{prop:randalg}, which solves the search version, applies to the
decision version as well, while our hardness results, which consider the
decision version, apply to the search version as well.

\paragraph{Promise CSPs}
\label{sec:pcsp}

We describe how the SetSAT problem fits into the general framework of promise CSPs (PCSPs). For a more in-depth algebraic study of PCSPs, we refer the reader to~\cite{bbko19}.

A \emph{relational structure} $\A$ is a tuple $(A; R_1,\ldots,R_m)$ where each
$R_i$ is a relation on $A$. We say that two relational structures are
\emph{similar} if their relations have the same sequence of arities. A
\emph{homomorphism} between similar relational structures $\A=(A;
R_1^\A,\ldots,,R_m^\A)$ and $\B=(B; R_1^\B,\ldots,R_m^\B)$ is a function $h:A
\to B$ such that $(a_1,\ldots,a_{\ar(R^\A_i)}) \in R^\A_i$ implies
$(h(a_1),\ldots,h(a_{\ar(R^\A_i)})) \in R^\B_i$ for all $i$. We denote this by
$\A \to \B$. 

\begin{defi}\label{defi:pcsp}
Let $(\A,\B)$ be a pair of similar relational structures such that there is a homomorphism $\A \to \B$. The pair $(\A,\B)$ is called the \emph{template} of the \emph{promise constraint satisfaction problem} $\PCSP(\A,\B)$. The decision version of $\PCSP(\A,\B)$ is as follows: given as input a relational structure $\mathbf{C}$ similar to $\A$ and $\B$, decide whether $\mathbf{C}$ admits a homomorphism to $\A$, or does not even admit a homomorphism to $\B$. The \emph{promise} is that it is never the case that $\mathbf{C} \to \B$ but $\mathbf{C} \not\to \A$. The search problem asks to find a homomorphism $\C \to \B$, given that there exists a homomorphism $\C \to \A$. 
\end{defi}

Since $(a,g,k)$-SetSAT is a PCSP where all relations have fixed arity $k$, it is
possible to transform SetSAT instances from their CNF representation into the
PCSP representation of Definition~\ref{defi:pcsp}. For domain size $d$ and set size $s$, there are $\binom{d}{s}$ different literals, and therefore $L:=\binom{d}{s}^k$ different types of clauses. Suppose that $f$ enumerates the types of clauses. We can represent each SetSAT instance $\Psi$ as a
relational structure $\C=(C; R^\C_1,\ldots,R^\C_L)$, where
$C=\{x_1,\ldots,x_n\}$ is the set of variables appearing in $\Psi$ and $R^\C_i$
is a $k$-ary relation corresponding to the clause $f(i)$. For each clause
$(S_1(x_1) \vee \ldots \vee S_k(x_k))$ of type $f(i)$ in $\Psi$, we add the
tuple $(x_1,\ldots,x_k)$ to $R^\C_i$, so that each $R^\C_i$ collects the tuples of variables appearing in clauses of type $f(i)$.

Now define $R^\A_i$ (respectively $R^\B_i$) to be the $k$-ary relation over
$[d]$ containing $(a_1,\ldots,a_k)$ if and only if $(a_1,\ldots,a_k)$
$g$-satisfies (respectively $a$-satisfies) the clause $f(i)$ when the variable
of the $j$-th literal of $f(i)$ is set to $a_j$, for $1\leq j \leq k$. Let
$\A=([d],R^\A_1,\ldots,R^\A_L)$ and $\B=([d],R^\B_1,\ldots,R^\B_L)$.
Then $(a,g,k)$-SetSAT is precisely $\PCSP(\A,\B)$: the identity function is a homomorphism from
$\A$ to $\B$, a homomorphism $\C \to \A$
represents a $g$-satisfying assignment to $\Psi$, and a homomorphism $\C \to \B$
represents an $a$-satisfying assignment to $\Psi$.

\paragraph{Polymorphisms}

The following concept from the algebraic study of PCSPs is central to our hardness result.

Let $f:A^m\to B$ be a function. We say that $f$ is a \emph{polymorphism} of the
template $(\A,\B)$ if, for $\bar{a}^1,\ldots,\bar{a}^m\in R^\A_i$, we have that
$f(\bar{a}^1,\ldots,\bar{a}^m)\in R^\B_i$; here $f$ is applied componentwise. We
will denote by $\Pol(\A,\B)$ the set of all polymorphisms of the template
$(\A,\B)$.
For every template, trivial polymorphisms are given by \emph{dictators},
which are functions $p$ of the form $p(x_1,\ldots,x_m)=f(x_i)$,
where $f$ is a homomorphism from $A$ to $B$.

In particular, $f:[d]^m\to[d]$ is a polymorphism of $(a,g,k)$-SetSAT if for every
SetSAT clause $C$ of width $k$ and for every tuple
$\bar{v}^1,\ldots,\bar{v}^m \in [d]^k$ of $g$-satisfying assignments to
$C$, we have that $f(\bar{v}^1,\ldots,\bar{v}^m)$  is an $a$-satisfying assignment to $C$.

\section{Tractability}
\label{sec:tractability}

How big must one make the fraction of satisfied literals in order for the SetSAT problem to become tractable? The following proposition shows that $\frac{s}{s+1}$ is sufficient. 

\begin{prop}\label{prop:randalg}
For $1\leq s < d$ and $\frac{g}{k} \geq \frac{s}{s+1}$, $(1,g,k)$-SetSAT is solvable in expected polynomial time.
\end{prop}

\begin{algorithm}
\caption{Randomised algorithm for $(1,g,k)$-SetSAT with $\frac{g}{k} \geq \frac{s}{s+1}$}\label{randalg}
\begin{algorithmic}[1]
\State $x \gets \text{arbitrary assignment}$
\While{$x$ does not satisfy input formula $\Psi$}
\State Arbitrarily pick a falsified clause $C$
\State Randomly choose from $C$ a literal $S(x_i)$ 
\State\label{satlit} Randomly choose a value for $x_i$ so that $S(x_i)$ is satisfied 
\EndWhile
\Return $x$

\end{algorithmic}
\end{algorithm}

\begin{proof}

Algorithm~\ref{randalg} finds a satisfying assignment to a $g$-satisfiable formula in expected polynomial time. The algorithm and its analysis are based on~\cite[Proposition~6.1]{agh17}, which in turn is based on Papadimitriou's randomised algorithm for 2-SAT~\cite{pap1991}.  

Suppose that $\Psi$ has a $g$-satisfying assignment $x^*$. Let $x^t$ be the
  assignment obtained in iteration $t$ of the algorithm, and let
  $D_t=\text{dist}(x^t,x^*)$, where dist$(x,y)$ is the Hamming distance between
  $x$ and $y$. Since $D_{t}-D_{t-1} \in \{-1,0,1\}$ for every $t$,\footnote{In
  the special case $d=2$, we have $D_{t}-D_{t-1}\in\{-1,1\}$.} we have 

\begin{align*}
\mathbb{E}(D_t - D_{t-1}) &=  \mathbb{P}(D_t - D_{t-1} = 1) -\mathbb{P}(D_t - D_{t-1} = -1)\\
&\leq \frac{k-g}{k} - \frac{g}{k}\frac{1}{s} \leq 0 \text{\quad if and only if \quad} \frac{g}{k} \geq \frac{s}{s+1}.
\end{align*}

The sequence $D_0,D_1,\ldots$ is a random walk starting between $0$ and $n$ with
each step either unbiased or biased towards $0$. This is a ``gambler's ruin''
chain with reflecting barrier (because the distance cannot increase beyond
$n$). With constant probability, such a walk hits $0$ (``the gambler is
broke'') within $n^2$ steps and the probability that the algorithm fails to
find a satisfying assignment within $crn^2$ steps is at most $2^{-r}$ for some
constant $c$. 
\end{proof}

\begin{rmk}
	The proof of Proposition~\ref{prop:randalg} can be modified to show that Algorithm~\ref{randalg} also finds a satisfying assignment when each literal corresponds to a set of size \emph{at most} $s$. This makes sense intuitively, as smaller literals give the algorithm a better chance of setting $x_i$ equal to $x^*_i$ in Step~\ref{satlit}.
\end{rmk}

We show that if $\frac{g}{k}\geq\frac{s}{s+1}$ then $(1,g,k)$-SetSAT has
a specific family of polymorphisms that leads to a \emph{deterministic} algorithm based
on linear programming.

A function $f:A^m \to B$ is \emph{symmetric} if $f(a_1,\ldots,a_m)=f(a_{\pi(1)},\ldots,a_{\pi(m)})$ for all $a_1,\ldots,a_m \in A$ and all permutations $\pi$ on $[m]$.

\begin{defi}
A symmetric function $f:[d]^m \to [d]$ is a \emph{plurality} if 
  \[f(x_1,\ldots,x_m)=\text{argmax}_{a \in [d]}\{ \text{\# of occurrences of
  }a\text{ in }(x_1,\ldots,x_m)\},\]
with ties broken in such a way that $f$ is symmetric.
\end{defi}

When $d=s+1$, all polymorphisms of SetSAT are
\emph{conservative}; i.e., they always return one of their input values, as the following proposition shows. This will be important when showing hardness, and also has implications for solvability by linear programming algorithms.

\begin{prop}\label{conspoly}
If $d=s+1$, then all polymorphisms of $(1,g,k)$-SetSAT are conservative.
\end{prop}

\begin{proof}
Let $f:[d]^m \to [d]$ be such that $f(a_1,\ldots,a_m) = b$ and $b \notin
  \{a_1,\ldots,a_m\}$. If $S$ is a literal not containing $b$, then the clause
  $(S(x_1) \vee \ldots \vee S(x_k))$ is $g$-satisfied (even $k$-satisfied) by setting all $x_i$ equal to any one of the $a_j$. Thus taking the $m$ assignments $(x_1=\dots=x_k=a_j)_{1 \leq j \leq m}$ and applying $f$ to each component, we get the assignment $x_1 = \dots = x_k = b$ which clearly does not 1-satisfy the clause, and so $f$ cannot be a polymorphism.
\end{proof}

\begin{prop}\label{polyarity}
Let $1 \leq s < d$. If $\frac{g}{k} > \frac{s}{s+1}$ then every plurality 
  function is a polymorphism of $(1,g,k)$-SetSAT. If $\frac{g}{k} = \frac{s}{s+1}$
  then every plurality function of arity $m \not\equiv 0 \mod s+1$ is
  a polymorphism of $(1,g,k)$-SetSAT, and no symmetric function of arity $m \equiv 0 \mod s+1$ is a polymorphism of $(1,g,k)$-SetSAT if we also have $d=s+1$.
\end{prop}

\begin{proof}
Let $f$ be a plurality function of arity $m$. Given $m$ $g$-satisfying assignments to a clause of width $k$, we are guaranteed to have at least $mg$ satisfying values among the $mk$ total values. Therefore there is a coordinate $i$, $1\leq i \leq k$, containing at least $\left \lceil \frac{mg}{k} \right\rceil$ satisfying values, that is, at least $\left \lceil \frac{mg}{k} \right\rceil$ values not equal to any of the values $b_1,\ldots,b_{d-s}$ forbidden by the $i$-th literal of the clause. For $f$ to be a polymorphism, we require that the most frequent satisfying value in coordinate $i$ appears more often than all the $b_1,\ldots,b_{d-s}$ combined, for which it suffices that $\left\lceil \frac{mg}{k} \right\rceil/s > m-\left\lceil \frac{mg}{k} \right\rceil$. This is equivalent to $\left\lceil \frac{mg}{k} \right\rceil > \frac{s}{s+1}m$, which follows from $\frac{g}{k} > \frac{s}{s+1}$, so $f$ is a polymorphism.

In the case $\frac{g}{k} = \frac{s}{s+1}$, the same argument works so long as $\frac{mg}{k}$ is not an integer, since by taking the ceiling we obtain a value strictly greater than $\frac{s}{s+1}m$. Since $\frac{g}{k} = \frac{s}{s+1}$, we have $\frac{g}{k}m = \frac{s}{s+1}m$ and this is an integer only if $m$ is a multiple of $s+1$. 

To show that there are no symmetric polymorphisms when $\frac{g}{k}=\frac{s}{s+1}$, $m$ is a multiple of $s+1$, and $d=s+1$, note that $\frac{g}{k}=\frac{s}{s+1}$ implies that $k$ is divisible by $s+1$. Let $M$ be the $(s+1) \times (s+1)$ matrix whose first row is $12\cdots s+1$ and whose $i$-th row for $2 \leq i \leq s+1$ is obtained from the $(i-1)$-st row by shifting it cyclically to the left by one coordinate. We stack $ \frac{k}{s+1} $ copies of $M$ on top of each other and take $\frac{m}{s+1}$ copies of this stack side-by-side to form the $k \times m$ matrix $M'$. If $f$ is symmetric, it returns the same value $b \in \{1,\ldots,s+1\}$ when applied to each row of $M'$. Every column of $M'$ satisfies exactly an $\frac{s}{s+1}$-fraction of the literals in the clause whose $k$ literals are all $S_{\{1,\ldots,s+1\} \setminus \{b\}}$.
On the other hand, the assignment produced by applying $f$ to each row of $M'$ does not even 1-satisfy this clause, so $f$ is not a polymorphism.
\end{proof}

Proposition~\ref{polyarity} has interesting consequences for solvability of
$(1,g,k)$-SetSAT via convex relaxations. By~\cite[Theorem~7.9]{bbko19},
$(1,g,k)$-SetSAT is solvable by the basic linear programming relaxation if
$\frac{g}{k}>\frac{s}{s+1}$ (since there exist symmetric polymorphisms of all
arities) but not solvable by the basic linear programming relaxation if
$\frac{g}{k}=\frac{s}{s+1}$ and $d=s+1$ (since there do not exist symmetric polymorphisms of
all arities). By~\cite[Theorem~3.1]{bg20}, $(1,g,k)$-SetSAT is solvable by the
combined basic linear programming and affine integer programming relaxation if
$\frac{g}{k}\geq\frac{s}{s+1}$ (since there exist symmetric polymorphisms of
infinitely many arities). We note that iterative rounding of the basic linear
programming relaxation could also be used to get a deterministic algorithm as
in~\cite{agh17}.

\section{Layered Label Cover and smug sets}
\label{sec:labelcover}

In this section we define a variant of the Label Cover problem, which reduces the task of showing hardness to showing that all polymorphisms satisfy a certain combinatorial property, and then in Section~\ref{sec:smallsmug}, we show that the polymorphisms of SetSAT satisfy this property. From now on we assume that $d=s+1$, as in Theorem~\ref{thm:main}.

An \emph{$\ell$-Layered Label Cover} instance is a sequence of $\ell+1$ sets $X_0,\dots,X_\ell$ (called \emph{layers}) of variables with range $[m]$, for some \emph{domain size} $m\in\mathbb{N}$, and a set of constraints $\Phi$.
Each constraint is a function (often called a projection constraint) from a variable $x \in X_i$ to a variable in a further layer $y \in X_j$, $i < j$: that is, a function denoted $\phi_{x\to y}$ which is satisfied by an assignment $\sigma\colon X_0 \cup \dots \cup X_\ell \to [m]$ if $\sigma(y) = \phi_{x\to y}(\sigma(x))$.
A \emph{chain} is a sequence of variables $x_i \in X_i$ for $i=0,\dots,\ell$ such that there are constraints $\phi_{x_i \to x_j}$ between them, for $i < j$.
A chain is \emph{weakly satisfied} if at least one of these constraints is satisfied.

The basis for our hardness result is the hardness of distinguishing fully satisfiable instances from those where no constant fraction of chains can be weakly satisfied.
This follows by a simple adaptation of a reduction from the work of Dinur,
Guruswami, Khot, and Regev~\cite{DinurGKR05}.

\begin{thm}\label{thm:layeredGLC}
	For every $\ell \in \mathbb{N}$ and $\varepsilon>0$, there is an $m\in\mathbb{N}$ such that
	it is NP-hard to distinguish $\ell$-Layered Label Cover instances with domain size $m$ that are fully satisfiable
	from those where not even an $\varepsilon$-fraction of all chains is weakly satisfied.
\end{thm}

\begin{proof}
	For $\ell=1$ a chain consists of just one constraint,
	so weakly satisfying the chain is the same as satisfying its constraint.
	The claim is then equivalent to the hardness of the standard Bipartite Gap Label Cover problem,
	which holds even for \emph{bi-regular} instances: that is, instances $(Y,Z)$ such that
	every variable in $Y$ occurs in constraints with exactly $d_{+}$ variables in $Z$ and
	every variable in $Z$ occurs in constraints with exactly $d_{-}$ variables in $Y$,
	for some $d_{+},d_{-} \in \mathbb{N}$. (This hardness follows from the PCP
	theorem~\cite{Arora1998:jacm-proof,Arora98:jacm} and Raz's parallel repetition
  theorem~\cite{Raz98:sicomp}, cf.~\cite{DinurGKR05} and~\cite{Arora09:book}.)
	
	For $\ell>1$ we reduce from a bi-regular instance of Bipartite Gap Label Cover with variable sets $Y$ and $Z$, domain size $m$, constraints $\Gamma$ and gap $\varepsilon' := \varepsilon/\binom{\ell+1}{2}$.
	Let the domain size of the constructed instance $\Phi$ be $m^\ell$.
	Let the variable sets be $X_i := Z^i \times Y^{\ell-i}$ for $i=0,\dots,\ell$
	(that is, $\ell$-tuples of $i$ variables from $Z$ followed by $\ell-i$ variables from $Y$;
	this makes indices notationally more convenient than the other way around).
	Let the constraints between $X_i$ and $X_{j}$ (for $0\leq i<j \leq \ell$) be defined for pairs of tuples $\bar{x}$ and $\bar{x}'$ of the form:
	\begin{align*}
	\bar{x}=(z_1,\dots,z_i,\ \ &y_{i+1},\dots,y_j,\ \ y_{j+1},\dots,y_\ell)\in X_i
	\text{\quad and}\\
	\bar{x}'=(z_1,\dots,z_i,\ \ &z_{i+1},\dots,z_j,\ \ y_{j+1},\dots,y_\ell)\in X_j
	\end{align*}
	such that the original instance has a constraint $\phi_{y_{k}\to z_{k}} \in \Gamma$ for $k= i+1,\dots,j$.
	Let the new projection constraint $\phi_{\bar{x}\to\bar{x}'}$ map $(a_1,\dots,a_\ell)$ to $(b_1,\dots,b_\ell)$
	where $b_k := \phi_{y_k\to z_k}(a_k)$ for $k= i+1,\dots,j$ and $b_k := a_k$ otherwise.
	This concludes the construction.
	
	Note that chains in this instance are in bijection with $\ell$-tuples of original constraints in $\Gamma$.
	Indeed, a chain $\bar{x}_i \in X_i$ ($i=0,\dots,\ell$) is determined by $\bar{x}_0=(y_1,\dots,y_\ell)$ and $\bar{x}_\ell=(z_1,\dots,z_\ell)$
	such that $\Gamma$ has constraints $\phi_{y_k\to z_k}$ for $k=1,\dots,\ell$.
	Moreover, for each $i<j$, every constraint $\phi_{\bar{x}\to\bar{x}'}$ between $\bar{x}\in X_i$ and $\bar{x}'\in X_j$ appears in the same number of chains
	(namely $d_{-}^{\ \,i}\cdot d_{+}^{\ \ell-j}$).
	
	If the original instance $\Gamma$ was fully satisfiable then so is the new one $\Phi$:
	indeed, if $\sigma$ is a satisfying assignment for $\Gamma$,
	then $\bar{x} \mapsto(\sigma(x_1),\dots,\sigma(x_\ell))$ is a satisfying assignment for $\Phi$.
	
	Suppose now that in $\Phi$, an assignment $\sigma\colon X_0\cup\dots\cup X_\ell \to [m]^\ell$
	weakly satisfies at least $\varepsilon$ of all chains.
	Then there exists $0\leq i < j \leq \ell$ such that
	at least $\varepsilon/\binom{\ell+1}{2}=\varepsilon'$ of all chains
	are weakly satisfied at a constraint between $X_i$ and $X_j$.
	Every constraint between $X_i$ and $X_j$ is contained in the same number of chains, say $C$,
	hence at least $\varepsilon'$ of the constraints between $X_i$ and $X_j$ are satisfied
	(indeed, the number of thus satisfied chains is exactly $C$ times the number of satisfied constraints; similarly, the number of all chains is exactly $C$ times the number of all constraints between $X_i$ and $X_j$).
	
	Choose an arbitrary coordinate $k$ in $i+1,\dots,j$.
	Partition $X_i$ into equivalence classes such that $\bar{x},\bar{x}'$ are in the same class if they are identical on all coordinates except possibly coordinate $k$.
	Partition $X_j$ in the same way.
	There exists a pair of classes between which constraints exist and at least $\varepsilon'$ of them are satisfied.
	That is, there are
	\begin{align*}
	x_1,\dots,x_{k-1},\ \ &x_{k+1},\dots,x_\ell \in Y \cup Z\text{\quad and}\\
	x'_1,\dots,x'_{k-1},\ \ &x'_{k+1},\dots,x'_\ell \in Y \cup Z
	\end{align*}	
	such that $\sigma$ satisfies at least $\varepsilon'$ of the constraints between pairs of the form
	\begin{align*}
	(x_1,\dots,x_{k-1},\ y,\ x_{k+1},\dots,x_\ell) \in X_i \quad\ \quad\\
	(x'_1,\dots,x'_{k-1},\ z,\ x'_{k+1},\dots,x'_\ell) \in X_j \quad\ \quad
	\end{align*}
	where a constraint $\phi_{y\to z}$ exists in $\Gamma$.
	Therefore, one can define an assignment $\sigma'\colon Y\cup Z \to [m]$ 
	by letting $\sigma'(y)$ and $\sigma'(z)$ be the $k$-th element of the value in $[m]^\ell$ resulting from applying $\sigma$ to the above tuples, respectively for $y \in Y$ and $z \in Z$.
	This assignment then satisfies at least $\varepsilon'$ of all the constraints $\phi_{y\to z}$ of the original instance $\Gamma$.
\end{proof}

In order to use Theorem~\ref{thm:layeredGLC} to derive hardness for PCSPs,
we use the algebraic approach:
every PCSP is equivalent to a promise problem about satisfying minor conditions with polymorphisms.
We give definitions first, following~\cite{bbko19}, to where we refer the reader for a more
detailed exposition.

For $f\colon A^n \to B$, $g\colon A^m \to B$, and $\pi\colon [n]\to [m]$,
we say that $g$ is the \emph{minor} of $f$ obtained from $\pi$ if
\begin{equation}\label{eq:minor}
  g(x_1,\dots,x_m) \approx f(x_{\pi(1)}, \dots, x_{\pi(n)}),
\end{equation}
where $g\approx f$ means that the values of $g$ and $f$ agree on every
input in $A^m$.
We write $f \xrightarrow{\pi} g$ as a shorthand for \eqref{eq:minor}.
For $\pi\colon [n]\to [m]$, the expression $f \xrightarrow{\pi} g$  is called a \emph{minor identity}.

A \emph{minion} on a pair of sets $(A,B)$ is a non-empty set of functions from $A^n$ to $B$ (for $n \in \mathbb{N}$)
that is closed under taking minors.

A \emph{bipartite minor condition} is a finite set $\Sigma$ of minor identities
where the sets of function symbols used on the left- and right-hand sides are
disjoint. More precisely, $\Sigma$ is
a pair of disjoint sets $U$ and $V$ of function symbols of arity $n$ and $m$, respectively,
and a set of minor identities of the form $f \xrightarrow{\pi} g$,
where $g\in U$, $f\in V$ and $\pi:[n]\to[m]$.
A bipartite minor condition $\Sigma$ is \emph{satisfied} in a minion $\Mm$
if there is an assignment $\xi:U\cup V\to\Mm$ that assigns
to each function symbol a function from $\Mm$ of the corresponding arity so that
for every identity $f \xrightarrow{\pi} g$ in $\Sigma$,
we have $\xi(f) \xrightarrow{\pi} \xi(g)$ in $\Mm$.
A bipartite minor condition is called \emph{trivial} if it is satisfied in every
minion, or equivalently, in the minion consisting of all projections on $\{0,1\}$.
Since choosing a projection of arity $n$ is the same as choosing an element of
$[n]$, deciding whether a bipartite minor condition is trivial is the same as
the standard Label Cover~\cite{Arora09:book}.

We can now define the \emph{promise satisfaction of a minor condition}
problem. For a minion $\Mm$ and an integer $m$, $\PMC_\Mm(m)$ is the following promise problem: given a bipartite minor
condition $\Sigma$ that involves only symbols of arity at most $m$, the answer
should be YES if $\Sigma$ is trivial and \textsc{NO} if $\Sigma$ is not satisfiable in
$\Mm$ (the promise is that either of those two cases holds, i.e. an algorithm can behave arbitrarily otherwise).
Barto et al.~\cite{bbko19} show that $\PCSP(\A,\B)$ is log-space equivalent to $\PMC_\Mm(m)$, for $\Mm = \Pol(\A,\B)$ and $m$ a constant depending on $\A$ only.

A final piece of notation before we prove a corollary of Theorem~\ref{thm:layeredGLC}. 
A \emph{chain of minors} is a sequence of the form $f_0 \xrightarrow{\pi_{0,1}} f_1 \xrightarrow{\pi_{1,2}} \dots \xrightarrow{\pi_{\ell-1,\ell}} f_\ell$.
We shall then write $\pi_{i,j} \colon [\ar(f_i)]\to[\ar(f_j)]$ for the composition of $\pi_{i,i+1}, \dots, \pi_{j-1,j}$, for any $0\leq i < j\leq \ell$.
Note that $f_i \xrightarrow{\pi_{i,j}} f_j$.

\begin{cor}[of Theorem~\ref{thm:layeredGLC}]
  \label{cor:sel}
	Let $\Mm$ be a minion.
	Suppose there are constants $k,\ell\in\mathbb{N}$ and an assignment of a set of at most $k$ coordinates $\sel(f)\subseteq[\ar(f)]$ to every $f\in\Mm$ such that for every chain of minors $f_0 \xrightarrow{\pi_{0,1}} f_1 \xrightarrow{\pi_{1,2}} \dots \xrightarrow{\pi_{\ell-1,\ell}} f_\ell$,
	there are $0\leq i < j \leq \ell$ such that $\pi_{i,j}(\sel(f_i))\cap \sel(f_j) \neq \emptyset$.
	Then $\PMC_\Mm(m)$ is NP-hard, for $m$ large enough.
  In particular, if $\Mm = \Pol(\A,\B)$, then $\PCSP(\A,\B)$ is NP-hard.
\end{cor}
\begin{proof}[Proof]
	For $\ell,k$ as in the assumption, let $\varepsilon := \frac{1}{k^2}$ and let $m$ be as given by Theorem~\ref{thm:layeredGLC}.
	We reduce an $\ell$-Layered Label Cover instance by replacing each variable $x$ with a symbol $f_x$ of arity $m$ and each constraint $\phi_{x \to y} \colon [m] \to [m]$ by the minor condition $f_x \xrightarrow{\phi_{x\to y}} f_y$.
	If the original instance was fully satisfiable, the new instance is
  trivial (i.e., fully satisfiable by projections).
	
	If the constructed instance is satisfied by functions in the minion $\Mm$, we define an assignment to the original instance by selecting, for each variable $x$, a random coordinate from $\sel(f_x)\subseteq [m]$ (uniformly, independently).
	The assumption guarantees a set of constraints $\phi_{x\to y}$ such that (1)~each chain contains at least one and (2) for each such constraint  $\phi_{x\to y}$, we have $\phi_{x\to y}(\sel(f_x))\cap \sel(f_{y}) \neq \emptyset$.
	The random choice then satisfies each of these constraints, and hence weakly satisfies each chain, with probability at least $\frac{1}{k^2}=\varepsilon$.
	The expected fraction of weakly satisfied chains is thus at least $\varepsilon$ and a standard maximisation-of-expectation procedure deterministically finds an assignment which certifies this.
\end{proof}

We now introduce the combinatorial property of polymorphisms which is crucial for our results.

\begin{defi}
	For a function $f\colon A^{\ar(f)} \to B$ we say that a set of coordinates $S \subseteq [\ar(f)]$ is a \emph{smug} set if there is an input vector $\bar{v}\in A^{\ar(f)}$ such that $S=\{i \mid v_i = f(\bar{v})\}$.
\end{defi}

Since $d=s+1$, Proposition~\ref{conspoly} applies and so for every polymorphism $f$ of SetSAT and every $\bar{v}$, the set $\{i \mid v_i = f(\bar{v})\}$ is nonempty. The following result connects Layered Label Cover to smug sets and will be used to prove hardness of $(1,g,k)$-SetSAT.

\begin{cor}\label{cor:smug}
	Let $\Mm$ be a minion.
	Suppose there are constants $k,\ell \in \mathbb{N}$ such that the following holds, for every $f\in \Mm$:
	\begin{itemize}
		\item $f$ has a smug set of at most $k$ coordinates,
		\item $f$ has no family of more than $\ell$ (pairwise) disjoint smug sets,
		\item if $f\xrightarrow{\pi}g$ and $S$ is a smug set of $g$, then $\pi^{-1}(S)$ is a smug set of $f$.
	\end{itemize}
	Then $\PMC_\Mm(m)$ is NP-hard, for $m$ large enough.
  In particular, if $\Mm = \Pol(\A,\B)$, then $\PCSP(\A,\B)$ is NP-hard.
\end{cor}
\begin{proof}
	For each $f \in \Mm$, we define $\sel(f)$ as a smug set of at most $k$ coordinates, arbitrarily chosen (some such set exists by the first condition).
	Consider a chain $f_0 \xrightarrow{\pi_{0,1}} f_1 \xrightarrow{\pi_{1,2}} \dots \xrightarrow{\pi_{\ell-1,\ell}} f_\ell$.
	Suppose to the contrary that for each $0\leq i < j \leq \ell$, $\pi_{i,j}(\sel(f_i))$ is disjoint from $\sel(f_{j})$, or equivalently, that $\sel(f_i)$ is disjoint from $\pi_{i,j}^{\ -1}(\sel(f_{j}))$.
	This implies that $\pi_{0,i}^{\ -1}(\sel(f_i))$ is disjoint from $\pi_{0,i}^{\ -1}(\pi_{i,j}^{\ -1}(\sel(f_{j}))) = \pi_{0,j}^{\ -1}(\sel(f_j))$.
	That is, the sets $\pi_{0,i}^{\ -1}(\sel(f_i))$ for $i=0\dots \ell$ are pairwise disjoint.
	By the third condition they are smug sets of $f_0$.
	But by the second condition this is impossible.
\end{proof}

We note that in the proof of Corollary~\ref{cor:smug}, the exact definition of ``smug'' is irrelevant, as long as it satisfies the above three conditions.
	
It is easy to check that the definition of ``smug'' satisfies the third condition for any functions $f \xrightarrow{\pi} g$, not necessarily polymorphisms.
Indeed, if an input $\bar{v} \in A^{\ar(g)}$ to $g$ gives a smug set $S=\{j \mid v_j = g(\bar{v})\}$,
then the corresponding input $\bar{u} \in A^{\ar(f)}$ to $f$ defined as $u_i := v_{\pi(i)}$ satisfies $f(\bar{u})=g(\bar{v})$ and hence gives a smug set $\{i \mid u_i = f(\bar{u})\} = \{i \mid v_{\pi(i)} =g(\bar{v})\} = \{i \mid \pi(i) \in S\} = \pi^{-1}(S)$.

The definition of ``smug'' is particularly well-suited to our problem, because whether $f$ is a polymorphisms or not depends only on its family of smug sets. 

\begin{lem}\label{lem:smugPolym}
	Let $1\leq s$ and $1\leq g < k$.
	A function $f \colon [s+1]^m \to [s+1]$ is a polymorphism of $(1,g,k)$-SetSAT
	if and only if there is no multiset $S_1,\dots,S_k$ of smug sets of $f$,
	such that each coordinate $\ell\in [m]$ is contained in at most $k-g$ of them.
\end{lem}

Before setting out to prove Lemma~\ref{lem:smugPolym}, we give an example in Figure~\ref{fig:smug} of a function that is not a polymorphism of $(1,3,5)$-SetSAT with set size 2 and domain size 3.

\begin{figure}[h!]
	\centering
	\vspace*{-\baselineskip}
	\begin{tikzpicture}[yscale=-0.4, xscale=0.37]
		\def\ww{10};
		\draw[rounded corners=1] (0,0) rectangle ++(\ww, 5);
		\draw[vsmug] (0,0) rectangle ++(5, 1);
		\draw[vsmug] (2,1) rectangle ++(2, 1); \draw[vsmug] (6,1) rectangle ++(2, 1);
		\draw[vsmug] (4,2) rectangle ++(3, 1);
		\draw[vsmug] (8,3) rectangle ++(2, 1);
		\draw[vsmug] (8,4) rectangle ++(2, 1);
		\matrix (m) [matrix of math nodes,matrix anchor=m-1-1.center,column sep={0.96em,between origins},row sep={1.04em,between origins}] at (0.5,0.5) {
			3 & 3 & 3 & 3 & 3 & 1 & 1 & 2 & 1 & 2\\
			3 & 3 & 2 & 2 & 1 & 1 & 2 & 2 & 3 & 3\\
			3 & 3 & 2 & 2 & 1 & 1 & 1 & 2 & 3 & 3\\
			1 & 2 & 1 & 2 & 1 & 2 & 1 & 2 & 3 & 3\\
			1 & 2 & 1 & 2 & 1 & 2 & 1 & 2 & 3 & 3\\
		};
		\draw[<->] (-0.4,0) --node[auto,swap]{$k=5$} (-0.4,5);
		\draw[<->] (0,-0.4) --node[auto]{$m$} (\ww,-0.4);
		\node[anchor=west,label={above:$f\phantom{ 0}$}]
		                   at (\ww, 0.5) {$\longrightarrow 3$};
		\node[anchor=west] at (\ww, 1.5) {$\longrightarrow 2$};
		\node[anchor=west] at (\ww, 2.5) {$\longrightarrow 1$};
		\node[anchor=west] at (\ww, 3.5) {$\longrightarrow 3$};
		\node[anchor=west] at (\ww, 4.5) {$\longrightarrow 3$};
		\node[anchor=west] at (\ww, 5.7) {$\phantom{\longrightarrow {}}\bar{o}$};
		\begin{scope}[shift={(\ww+7,0)}]
			\node[anchor=south] at (0,-0.3) {clause};
			\node at (0,0.5) {\small$x_1 \neq 3$};
			\node at (0,1.5) {\small$x_2 \neq 2$};
			\node at (0,2.5) {\small$x_3 \neq 1$};
			\node at (0,3.5) {\small$x_4 \neq 3$};
			\node at (0,4.5) {\small$x_5 \neq 3$};
			\node at (-1.7,1.15) {\small$\vee$};
			\node at (-1.7,2.15) {\small$\vee$};
			\node at (-1.7,3.15) {\small$\vee$};
			\node at (-1.7,4.15) {\small$\vee$};
			\node at (0,5.7) {\small$\phantom{x_5 \neq {}}\bar{b}$};
		\end{scope}
	\end{tikzpicture}
	\vspace*{-\baselineskip}	
	\caption{Illustration of Lemma~\ref{lem:smugPolym}. Smug sets $S \subseteq [m]$ are highlighted in each row.}
	\label{fig:smug}
\end{figure}

\vspace*{-\baselineskip}
\begin{proof}[Proof of Lemma~\ref{lem:smugPolym}]
	
	A function $f \colon [s+1]^m \to [s+1]$ is not a polymorphism if and only if there is a 
	clause of the form $x_1 \neq b_1 \vee \dots \vee x_k \neq b_k$ (for some column vector $\bar{b} \in [s+1]^k$)
	and a sequence of $m$ column vectors $\bar{v}^{1},\dots,\bar{v}^m \in [s+1]^k$
	each of which $g$-satisfies the clause,
	but for which the vector $\bar{o} = f(\bar{v}^{1}, \dots, \bar{v}^{m})$ (with $f$ applied coordinatewise) does not even 1-satisfy the clause.
	The latter is equivalent to saying that $o_i = b_i$ for $i\in[k]$, that is,
	applying $f$ to the $i$-th row gives $f(v^1_i,\dots,v^m_i) = b_i$.
	The former is equivalent to saying that for each column $\bar{v}$ in $\bar{v}^1,\dots,\bar{v}^m$, the condition $v_i \neq b_i$ holds for at least $g$ indices $i\in [k]$ of that column.
	The two are hence equivalent to saying that for each column $\bar{v}^\ell$, $\ell\in[m]$,
	the condition $v^\ell_i = f(v^1_i,\dots,v^m_i)$ holds for at most $k-g$ indices $i \in [k]$ in that column.
	In other words,
	the $k$ row vectors $(v^1_i,\dots,v^m_i)$ for $i\in[k]$ have smug sets such that $\ell$ is contained in at most $k-g$ of these sets, for each coordinate $\ell \in [m]$.
\end{proof}

Checking the second condition for polymorphisms of our SetSAT problem is easy.
\begin{lem}\label{lem:disjoint}
	For every polymorphism $f$ of $(1,g,k)$-SetSAT,
	if $S_1,\dots,S_t$ are disjoint smug sets of $f$, then $t < \frac{k}{k-g}$.
\end{lem}
\begin{proof}
	Suppose to the contrary that $t \geq \frac{k}{k-g}$.
	Then we can build a multiset containing each $S_i$ up to ${k-g}$ times until we have exactly $k$ in total.
	We thus obtain a multiset of $k$ smug sets such that every coordinate is contained in at most $k-g$ of them.
\end{proof}

We have now shown that the second and third  conditions of Corollary~\ref{cor:smug} hold, so it remains only to show that SetSAT has small smug sets.

\section{Finding small smug sets}
\label{sec:smallsmug}

It is easy to show NP-hardness when $\frac{g}{k}\leq \frac{1}{2}$ (cf.
Proposition~\ref{padding}).
We now show a general reduction by finding a small smug set for
$(1,g,k)$-SetSAT whenever $\frac{g}{k}<\frac{s}{s+1}$. Again we assume $d=s+1$.

\begin{lem}\label{lem:sDisjoint}
	Let $f \colon [s+1]^m \to [s+1]$ be a polymorphism of $(1,g,k)$-SetSAT with set size $s$ and domain size $s+1$. There exists a smug set of $f$ of size at most $s-1$, or a family of $s$ disjoint minimal smug sets $S_1, \dots, S_{s}$.
\end{lem}
\begin{proof}
	Suppose that every smug set has size at least $s$.
	We show by induction on $t$ that there is a family of $t$ disjoint minimal smug sets $S_1,\dots,S_{t}$.
	Suppose we found $S_1,\dots,S_{t}$ for some $0 \leq t < s$ and we want to find $S_{t+1}$.
	Let $T$ be a set containing one arbitrary coordinate from each $S_i$, $i=1\dots t$.
	Let $\bar{v}\in [s+1]^m$ be the input vector with values $t+2$ on $T$, $i$ on $S_i \setminus T$ (for $i=1\dots t$) and $t + 1$ on the remaining coordinates $R := [m] \setminus (S_1 \cup \dots \cup S_t)$.
	Since $|T| \leq t < s$, $T$ is not smug, so $f(\bar{v}) \neq t+2$.
	By minimality, $S_i \setminus T$ are not smug for $i=1\dots t$, so $f(\bar{v}) \neq i$.
	Therefore, by conservativity of $f$ (Proposition \ref{conspoly}), the only remaining option is $f(\bar{v}) = t+1$.
	Thus $R$ is smug and disjoint from $S_i$.
	Taking $S_{t+1}$ to be a minimal smug set contained in $R$ proves the induction step.
\end{proof}

Together with Lemma~\ref{lem:disjoint}, Lemma~\ref{lem:sDisjoint} already
establishes (via Corollary~\ref{cor:smug}) NP-hardness when $s \geq \frac{k}{k-g} = \frac{g}{k-g}+1$ (equivalently, $\frac{g}{k} \leq \frac{s-1}{s}$):
 since there cannot be $s$ disjoint smug sets, every polymorphism has a smug set of size at most $s-1$.
The proof in the general case, when $\frac{g}{k}<\frac{s}{s+1}$, extends this approach by first finding (assuming there are no small smug sets) disjoint minimal smug sets $S_1,\dots,S_s$, then
exploiting the fact that each has a special coordinate whose removal makes it not smug,
and using these coordinates to find further variants of each $S_i$ with new special coordinates.

\begin{lem}\label{lem:smallSmug}
	Let $\frac{g}{k} < \frac{s}{s+1}$ and  
	let $f \colon [s+1]^m \to [s+1]$ be a polymorphism of $(1,g,k)$-SetSAT with set size $s$ and domain size $s+1$. Then $f$ has a smug set of size at most $g$.
\end{lem}
\begin{proof}
	Consider a polymorphism $f \colon [s+1]^m \to [s+1]$ of $(1,g,k)$-SetSAT.
	We prove by induction on $t$ that there is a smug set of size at most $t-1$,
	or there is a sequence of smug sets $S_1,\dots,S_t$ and a set $T$ such that
	(see Figure~\ref{fig:smallSmug}):
	\begin{enumerate}[label={(\roman*)}]
		\item $|T|=t$ and $|T \cap S_i| = 1$ for $i=1\dots t$ (hence $S_i \cap T \neq S_{i'} \cap T$ for $i \neq i'$);
		\item $S_i \setminus T $ is not smug for $i=1\dots t$;		
		\item $S_i \cap S_{i'} = \emptyset$ if $i \not\equiv i' \mod s$;
		\item $S_i \supseteq S_{i-s} \setminus T$ for $i > s$.		
	\end{enumerate}
	By Lemma~\ref{lem:sDisjoint} we can start with $t=s$ (by taking any $T$ containing one coordinate from each $S_i$).
	Suppose the above is true for $t \geq s$ and let us prove the same for $t+1$.
	If there is a smug set of size at most $t$ then we are done, so assume that $T$ is not smug.
	Let $\bar{v}\in [s+1]^m$ be the input vector with value $s+1$ on $T$ and different values from $\{1,\dots,s\}$ on $S_{t-i}\setminus T$ for $i=0 \dots s-2$ and on the set of remaining coordinates
	$R := [m] \setminus (S_t \cup \dots \cup S_{t-s+2} \cup T)$.
	Then by (ii), $R$ is smug.
	
	Observe that $R$ contains $S_{t-s+1} \setminus T$, because $S_t, \dots, S_{t-s+2},T$ are disjoint from that set by~(iii).
	We define $S_{t+1}$ to be a minimal subset of $R$ among smug sets containing $S_{t-s+1} \setminus T$.
	By (ii) $S_{t-s+1} \setminus T$ itself is not smug,
	so there exists some coordinate $\ell$ in $S_{t+1} \setminus S_{t-s+1}$.
	We choose it arbitrarily and set $T' := T \cup \{\ell\}$.
	
	We claim that the sequence of smug sets $S_1,\dots,S_{t+1}$ and the set $T'$ satisfy the above conditions.
	By minimality $S_{t+1} \setminus T'$ is not smug, so it satisfies (ii)
	and by definition it satisfies (iv).
	The set $S_{t+1}$ is disjoint from $S_t, \dots, S_{t-s+2},T$, because $R$ was.
	It is also disjoint from $S_i$ for $i \not\equiv t+1 \mod s$,  because for every such $i$, $S_i\setminus T$ is contained in one of $S_t, \dots, S_{t-s+2}$; this proves (iii).
	In particular $\ell$ is not contained in any of these sets,
	and since it is not contained in $S_{t-s+1}$,
	it is in fact not contained in any $S_i$ with $i < t+1$.
	Hence $|T'| = t+1$ and $|T' \cap S_i|  = |T \cap S_i| = 1$ for $i < t+1$.
	Clearly also $|T' \cap S_{t+1}| = |\{\ell\}| = 1$.
	Therefore, (i) is satisfied, concluding the inductive proof.
	
	\medskip
	Let us now consider sets as guaranteed above for $t=g+1$ (assuming there is no smug set of size at most $g$).
	Let $\bar{v}\in [s+1]^m$ be the input vector with value $i+1$ on $S_{t-i}\setminus T$ for $i=0 \dots s-1$,
	and	value $s+1$ on the remaining coordinates $R := \left([m] \setminus (S_t \cup \dots \cup S_{t-s+1}) \right) \cup T$.
	By (ii) the sets $S_{t-i} \setminus T$ are not smug, so $R$ is smug.
	We claim that the multiset obtained from $\{S_1,\dots,S_t\}$ by adding $(k-g-1)$ copies of the set $R$ contradicts Lemma~\ref{lem:smugPolym}:
	that is, each coordinate in $[m]$ is covered at most $k-g$ times by this multiset.
	
	Consider first the coordinates contained in $R$.
	By definition of $R$, they are disjoint from $S_{t-i} \setminus T$ for $i=0 \dots s-1$.
	By (iv), they are also disjoint from all sets $S_i\setminus T$ for $i=0\dots t$,
	because every such set is contained in one of the former.
	Hence if a coordinate in $R$ is also contained in one of $S_1,\dots,S_t$,
	then it is contained in $T$ and therefore in at most one of $S_1,\dots,S_t$, by (i).
	In total, it is thus covered at most $(k-g-1)+1 = k-g$ times.
	
	Consider now coordinates outside of $R$.
	By (iii), they can be covered only by sets $S_i$ with congruent indices $i \mod s$.
	Since $s > \frac{g}{k-g}$, we have $s(k-g) > g$, so there are $t = g+1 \leq s(k-g)$
	distinct indices in total in $\{1,\dots,t\}$.
	Hence at most $k-g$ of them can be pairwise congruent to each other $\hspace{-2mm} \mod s$.
	Thus coordinates outside of $R$ are also covered at most $k-g$ times.
\end{proof}

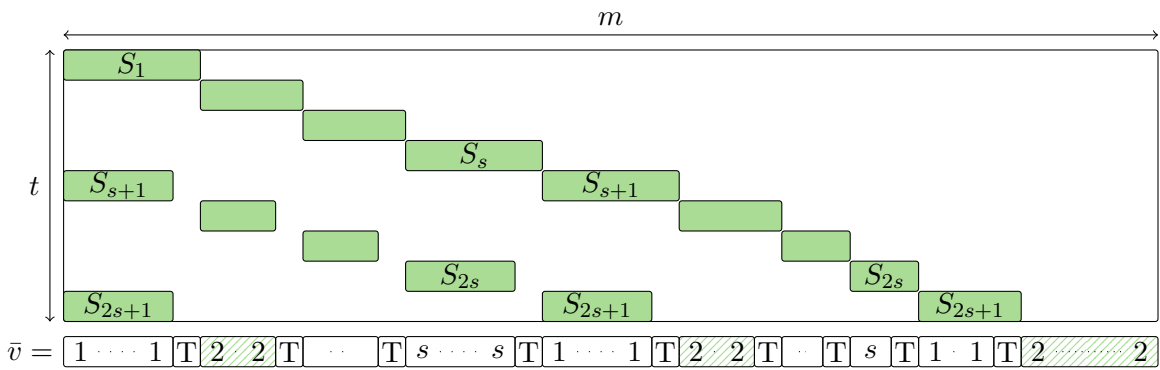
\begin{figure}[b!]
	\centering
	\def\ggap{0.8}
	\begin{tikzpicture}[yscale=-0.4, xscale=0.45]
	\draw[<->] (-0.4,0) --node[auto,swap]{$t$} (-0.4,9);
	\draw[<->] (0,-0.5) --node[auto]{$m$} (32,-0.5);	
	\draw [rounded corners=1] (0,0) rectangle ++(32, 9);
	\draw [rounded corners=1,fill=Csmug]
	(0,0)     rectangle ++(4, 1) node[midway] {$S_1$}
	++(0,0) rectangle ++(3, 1) node[midway] {}
	++(0,0) rectangle ++(3, 1) node[midway] {}
	++(0,0) rectangle ++(4, 1) node[midway] {$S_{s}$};
	\draw [rounded corners=1,fill=Csmug]
	(0,4)     rectangle ++(4-\ggap, 1) node[midway] {$S_{s+1}$}
	++(\ggap,0) rectangle ++(3-\ggap, 1) node[midway] {}
	++(\ggap,0) rectangle ++(3-\ggap, 1) node[midway] {}
	++(\ggap,0) rectangle ++(4-\ggap, 1) node[midway] {$S_{2s}$};
	\draw [rounded corners=1,fill=Csmug]
	(14,4)     rectangle ++(4, 1) node[midway] {$S_{s+1}$}
	++(0,0) rectangle ++(3, 1) node[midway] {}
	++(0,0) rectangle ++(2, 1) node[midway] {}
	++(0,0) rectangle ++(2, 1) node[midway] {$S_{2s}$};	
	\draw [rounded corners=1,fill=Csmug]
	(0,8)     rectangle ++(4-\ggap, 1) node[midway] {$S_{2s+1}$};
	\draw [rounded corners=1,fill=Csmug]
	(14,8)     rectangle ++(4-\ggap, 1) node[midway] {$S_{2s+1}$};
	\draw [rounded corners=1,fill=Csmug]
	(25,8)     rectangle ++(3, 1) node[midway] {$S_{2s+1}$};	
	
	\begin{scope}[shift={(0,9.5)}]
	\node[anchor=east] at (-0,0.5) {$\bar{v}=$};
	\draw [rounded corners=1] (0,0) rectangle ++(4-\ggap, 1);
	\node (a) at (0.5,0.5) {$1$}; \node (b) at (3.5-\ggap,0.5) {$1$}; \draw[loosely dotted] (a)--(b);
	\draw [rounded corners=1] (4-\ggap,0) rectangle ++(\ggap, 1);	
	\node at (4-0.5*\ggap,0.5) {T};	
	
	\draw [rounded corners=1,pattern=north east lines,pattern color=Csmug] (4,0) rectangle ++(3-\ggap, 1);
	\node (a) at (4.5,0.5) {$2$}; \node (b) at (6.5-\ggap,0.5) {$2$}; \draw[dotted] (a)--(b);
	\draw [rounded corners=1] (7-\ggap,0) rectangle ++(\ggap, 1);		
	\node at (7-0.5*\ggap,0.5) {T};
	
	\draw [rounded corners=1] (7,0) rectangle ++(3-\ggap, 1);
	\node (a) at (7.5,0.5) {}; \node (b) at (9.5-\ggap,0.5) {}; \draw[loosely dotted] (a)--(b);
	\draw [rounded corners=1] (10-\ggap,0) rectangle ++(\ggap, 1);	
	\node at (10-0.5*\ggap,0.5) {T};
	
	\draw [rounded corners=1] (10,0) rectangle ++(4-\ggap, 1);
	\node (a) at (10.5,0.5) {$s$}; \node (b) at (13.5-\ggap,0.5) {$s$}; \draw[loosely dotted] (a)--(b);
	\draw [rounded corners=1] (14-\ggap,0) rectangle ++(\ggap, 1);	
	\node at (14-0.5*\ggap,0.5) {T};	
	
	\draw [rounded corners=1] (14,0) rectangle ++(4-\ggap, 1);
	\node (a) at (14.5,0.5) {$1$}; \node (b) at (17.5-\ggap,0.5) {$1$}; \draw[loosely dotted] (a)--(b);	
	\draw [rounded corners=1] (18-\ggap,0) rectangle ++(\ggap, 1);	
	\node at (18-0.5*\ggap,0.5) {T};

	\draw [rounded corners=1,pattern=north east lines,pattern color=Csmug] (18,0) rectangle ++(3-\ggap, 1);
	\node (a) at (18.5,0.5) {$2$}; \node (b) at (20.5-\ggap,0.5) {$2$}; \draw[dotted] (a)--(b);
	\draw [rounded corners=1] (21-\ggap,0) rectangle ++(\ggap, 1);		
	\node at (21-0.5*\ggap,0.5) {T};
	
	\draw [rounded corners=1] (21,0) rectangle ++(2-\ggap, 1);
	\node (a) at (21.5,0.5) {}; \node (b) at (22.5-\ggap,0.5) {}; \draw[dotted] (a.center)--(b.center);
	\draw [rounded corners=1] (23-\ggap,0) rectangle ++(\ggap, 1);	
	\node at (23-0.5*\ggap,0.5) {T};	
	
	\draw [rounded corners=1] (23,0) rectangle ++(2-\ggap, 1);
	\node at (24-0.5*\ggap,0.5) {$s$};
	\draw [rounded corners=1] (25-\ggap,0) rectangle ++(\ggap, 1);	
	\node at (25-0.5*\ggap,0.5) {T};		

	\draw [rounded corners=1] (25,0) rectangle ++(3-\ggap, 1);
	\node (a) at (25.5,0.5) {$1$}; \node (b) at (27.5-\ggap,0.5) {$1$}; \draw[loosely dotted] (a)--(b);	
	\draw [rounded corners=1] (28-\ggap,0) rectangle ++(\ggap, 1);	
	\node at (28-0.5*\ggap,0.5) {T};
	
	\draw [rounded corners=1,pattern=north east lines,pattern color=Csmug] (28,0) rectangle ++(4, 1);
	\node (a) at (28.5,0.5) {$2$}; \node (b) at (31.5,0.5) {$2$}; \draw[dotted] (a)--(b);
	\end{scope}
	\end{tikzpicture}
	\caption{Illustration of smug sets obtained in the proof of Lemma~\ref{lem:smallSmug}.
	Each row represents one of the sets in the sequence $S_1,\dots,S_t$. The set
  $T$ is formed by coordinates with a T and get values $s+1$. The vector $\bar{v}$ is used to find the next row $S_{t+1}$.\vspace*{-\baselineskip}}
	\label{fig:smallSmug}
\end{figure}
	
This concludes the proof that smug sets satisfy the first condition of Corollary~\ref{cor:smug}
for polymorphisms of $(1,g,k)$-SetSAT with set size $s$ and domain size $s+1$,
assuming $\frac{g}{k} < \frac{s}{s+1}$.
Therefore, the problem is NP-hard and we have established the hardness part of
Theorem~\ref{thm:main}.

\section{The general case}
\label{sec:general}

In this section we will show how Theorem~\ref{thm:main} implies
Theorem~\ref{thm:general}. This amounts to arguing that a classification of
$(a,g,k)$-SetSAT can be obtained from the special case with $a=1$ and $d=s+1$.
We start with two easy reductions.

\begin{prop}\label{prop:shift}
	For any $1 \leq s < d$, the problems $(a,g,k)$-SetSAT and $(a+1,g+1,k+1)$-SetSAT are polynomial-time reducible to each other.
\end{prop}
\begin{proof}
	To reduce $(a,g,k)$-SetSAT to $(a+1,g+1,k+1)$-SetSAT, introduce a new variable $y$ and add $S(y)$ to each existing clause, where $S$ is any literal. If the original instance has a $g$-satisfying assignment, then the same assignment, extended by assigning $y$ to a value satisfying $S$, is a $(g+1)$-satisfying assignment to the new instance. Conversely, if the old instance is not $a$-satisfiable, then the new instance cannot be $(a+1)$-satisfiable, as each new clause contains at most one additional satisfied literal.
	
	In the other direction, from $(a+1,g+1,k+1)$-SetSAT to $(a,g,k)$-SetSAT, let $\Psi$ be the orignal instance. For each clause $C$ of $\Psi$ we make $k+1$ new clauses by taking all subsets of $k$ literals of $C$. If $\Psi$ has a $(g+1)$-satisfying assignment, then the same assignment is $g$-satisfying for the new instance since we have removed only one literal from each clause of $\Psi$. Conversely, if every assignment to $\Psi$ is not $(a+1)$-satisfying, then every assignment is at most $a$-satisfying. Removing one of the satisfied literals from a clause $C$ of $\Psi$ creates a new clause that is at most $(a-1)$-satisfiable. Therefore, in the new instance, every assignment is at most $(a-1)$-satisfying.
\end{proof}

\begin{prop}\label{domainsize}
	There is a polynomial-time reduction from $(1,g,k)$-SetSAT with set size $s$ and domain size $d$ to $(1,g,k)$-SetSAT with set size $s$ and domain size $d+1$.
\end{prop}
\begin{proof}
	The new instance produced by the reduction is the same as the old instance $\Psi$. If
	$\Psi$ is $g$-satisfiable, then it is again $g$-satisfiable by the same assignment and we ignore the new domain value. Conversely, a satisfying assignment over
	$[d+1]$ to $\Psi$ restricts to a satisfying
	assignment over $[d]$ by replacing $d+1$ with any value from $[d]$. This does not falsify any literals since all the literals of $\Psi$ range over $[d]$ only.
\end{proof}

\begin{proof}[Proof of Theorem~\ref{thm:general}] 
By Proposition~\ref{prop:shift}, we can assume $a=1$. The algorithm in
Proposition~\ref{prop:randalg} solves the problem in polynomial time as long as
$\frac{g}{k} \geq \frac{s}{s+1}$ (independent of $d$). Theorem~\ref{thm:main} states that $(1,g,k)$-SetSAT is NP-hard when $\frac{g}{k} < \frac{s}{s+1}$ and $d=s+1$. Proposition~\ref{domainsize} then extends this to larger $d$. 
\end{proof} 

We finish this section with proving the claim from Section~\ref{sec:intro} that literals described by sets of size less than $s$ can be emulated by literals of size exactly $s$.

\begin{prop}\label{setsize}
If $s \leq d-2$, there is a polynomial-time reduction from $(1,g,k)$-SetSAT with set size $s$ and domain size $d$ to $(1,g,k)$-SetSAT with set size $s+1$ and domain size $d$.
\end{prop}
\begin{proof}
	We replace each clause $S_1(x_1) \vee \dots \vee S_k(x_k)$
	with a set of $(d-s)^k$ clauses $S_1'(x_1) \vee \dots \vee S_k'(x_k)$,
	where $S_i'$ ranges over all supersets of $S_i$ of size $s+1$.
	Any $g$-satisfying assignment to the former clearly satisfies the latter.
	For a 1-satisfying assignment $\sigma$ to the latter,
	we claim that for every new clause $S_1'(x_1) \vee \dots \vee S_k'(x_k)$,
	at least one of the literals in $S_1(x_1) \vee \dots \vee S_k(x_k)$
	must be satisfied by $\sigma$. Suppose to the contrary that $\sigma(x_i)=a_i$ where 
	$a_i \not\in S_i$ for all $1\leq i \leq k$. Then the clause formed by the literals $S_i'=S_i \cup \{b_i\}$, where $b_i \in [d]\setminus (S_i \cup \{a_i\})$,
	would not be satisfied by $\sigma$, a contradiction. Note that as $d-s \geq 2$, the set $[d]\setminus (S_i \cup \{a_i\})$ is non-empty.
\end{proof}

\section{Impossibility results}
\label{sec:impos}

Here we show that $(1,g,k)$-SetSAT has rich polymorphisms (satisfying many non-trivial minor conditions),
even in the NP-hard range of parameters.
We thus demonstrate that certain sufficient condition for NP-hardness
from~\cite{bbko19} and earlier work cannot be
used to establish hardness of $(1,g,k)$-SetSAT for non-Boolean domains. We start with a few definitions.

Given an $n$-ary function $f:A^n\to B$, the first coordinate is called
\emph{essential} if there exist $a,a'\in A$ and $\bar{a}\in A^{n-1}$ such that
$f(a,\bar{a})\neq f(a',\bar{a})$; otherwise, the first coordinate is called
\emph{inessential}. Analogously, one defines the $i$-th
coordinate to be (in)essential.
The \emph{essential arity} of $f$ is the number of essential coordinates.
A minion has \emph{bounded essential arity} if there is some $k$
such that every function in the minion has essential arity at most $k$.

Let $\mathscr M $ and $\mathscr N$ be two minions.
A map $\xi:\mathscr M\to\mathscr N$ is called a \emph{minion
homomorphism} if (1) it preserves arities; i.e., maps $n$-ary functions to
$n$-ary functions, for all $n$; and (2) it preserves taking minors; i.e., for
each $\pi:[n]\to[m]$ and each $n$-ary $g\in\mathscr M$, we have
$\xi(g)(x_{\pi(1)},\ldots,x_{\pi(n)})=\xi(g(x_{\pi(1)},\ldots,x_{\pi(n)}))$.
A minion homomorphism $\mathscr M\to\mathscr N$ implies that
all minor conditions satisfied in $\mathscr M$ are also satisfied in $\mathscr N$;
this provides an algebraic way to give reductions between
PCSPs~\cite{bbko19}.
The basic example of this is the following theorem, based on the techniques used for promise SAT~\cite{agh17}.

\begin{thm}[\protect{\cite[Proposition~5.15]{bbko19}}]\label{thm:minbndarity}
  Let $(\A,\B)$ be a template. Assume that there exists a minion
  homomorphism $\xi:\Pol(\A,\B)\to\mathscr M$ for some minion $\mathscr M$ of bounded essential
  arity which does not contain a constant function (i.e., a function without
  essential coordinates). Then $\PCSP(\A,\B)$ is NP-hard.
\end{thm}

In fact, this follows from a slightly more general condition.

\begin{defi}
	Let $\epsilon > 0$. We say that a bipartite minor condition $\Sigma$ is $\epsilon$-\emph{robust} if no $\epsilon$-fraction of identities from $\Sigma$ is trivial (i.e. satisfiable by projections).
\end{defi}

\begin{thm}[\protect{\cite[Corollary~5.11]{bbko19}}]\label{thm:robust}
	If there exists an $\epsilon > 0$ such that $\Pol(\A,\B)$ does not satisfy any $\epsilon$-robust minor condition, then $\PCSP(\A,\B)$ is NP-hard.
\end{thm}

Theorem~\ref{thm:minbndarity} follows from Theorem~\ref{thm:robust} by observing that
minions of bounded essential arity cannot satisfy any sufficiently robust condition.
We give a proof for completeness.

\begin{lem}
	Let $\mathscr M$ be a minion where every function of arity $m$ has essential arity at most $f(m)$.
	Then $\mathscr M$ cannot satisfy any $\frac{1}{f(m)}$-robust bipartite minor condition involving symbols of arity at most~$m$.
\end{lem}
\begin{proof}
	Let $\Sigma$ be a $\frac{1}{f(m)}$-robust bipartite minor condition involving symbols of arity at most~$m$.
	Suppose $\Sigma$ is satisfied by $\mathscr M$,
	that is, there is an assignment $\xi$ from symbols in $\Sigma$ to functions in $\mathscr M$ of the same arity such that for every condition $f \xrightarrow{\pi} g$ in $\Sigma$ we have $\xi(f) \xrightarrow{\pi} \xi(g)$.
	Let $I(\xi(f))$ be the set of essential coordinates in $\xi(f)$. 
	It is easy to check that essential coordinates of a minor $\xi(g)$ of a function $\xi(f)$ correspond to essential coordinates of $\xi(f)$, that is:
	$\xi(f) \xrightarrow{\pi} \xi(g)$ implies $I(\xi(g)) \subseteq \pi(I(\xi(f)))$.
	Hence if we fix $\iota(g) \in I(\xi(g))$ arbitrarily, for each symbol $g$ on one side of $\Sigma$,
	and choose $\iota(f) \in I(\xi(f))$ uniformly at random, for each symbol $f$ on the other side of $\Sigma$, then for each condition $f \xrightarrow{\pi} g$ the corresponding condition $\pi(\iota(f)) = \iota(g)$ is satisfied with probability at least $\frac{1}{f(m)}$.
	Equivalently, replacing $\xi(g)$ with the projection to $\iota(g)$ and $\xi(f)$ with the projection $\iota(f)$, the condition $p_{\iota(f)} \xrightarrow{\pi} p_{\iota(g)}$ is satisfied
	with probability at least $\frac{1}{f(m)}$.	
	Therefore, there exists an assignment with projections that satisfies at least $\frac{1}{f(m)}$ of the conditions,
	which means $\Sigma$ is not $\frac{1}{f(m)}$-robust.
\end{proof}

We show that polymorphisms of SetSAT satisfy robust conditions and therefore the assumptions of Theorem~\ref{thm:minbndarity} and Theorem~\ref{thm:robust} are not met.
The same construction will also give polymorphisms excluding other approaches (e.g. polymorphisms without small ``fixing'' sets).
We first define how to reconstruct a polymorphism from a family of smug sets.

\begin{defi}\label{def:smugToPolym}
	Consider $(1,g,k)$-SetSAT with domain size $s+1$.
	Let $U$ be a finite set and let $\mathscr{S}=\{S_1,\dots,S_{|\mathscr{S}|}\}$ be a sequence of non-empty subsets of $U$ with the following properties:
	\begin{itemize}
		\item for every partition $U = U_1 \cup \dots \cup U_{s+1}$ into $s+1$ possibly empty sets,
		at least one of $U_1,\dots,U_{s+1}$ is in $\mathscr{S}$.
		\item for every $k$-tuple $(S_{i_1},\dots,S_{i_k}) \in \mathscr{S}^k$, some $u \in U$ is contained in at least $k-g+1$ of the $k$ sets.
	\end{itemize}
	Let $q_{\mathscr{S}} \colon [s+1]^{|U|} \to [s+1]$ be defined as follows.
	For an input $\bar{x} \in [s+1]^{|U|}$,
	partition the coordinates according to their value:
	that is, for $i \in [s+1]$ let $U_i := \{ u \in U \colon x_u = i \}$.
	Let $q_{\mathscr{S}}(\bar{x})$ be the value $i \in [s+1]$ such that $U_i \in \mathscr{S}$;
	if there are many such $i$, choose $U_i$ to be first in the sequence $\mathscr{S}$.
\end{defi}
By construction, all the smug sets of $q_{\mathscr{S}}$ are contained in $\mathscr{S}$.
By Lemma~\ref{lem:smugPolym}, $q_{\mathscr{S}}$ is a polymorphism.
Note that because of the preference for earlier sets in $\mathscr{S}$,
not all sets in $\mathscr{S}$ have to be smug, 
and there may exist different functions with the same family of smug sets.
On the other hand, the ordering in $\mathscr{S}$ matters only when comparing disjoint sets.

The following polymorphisms satisfy many non-trivial minor conditions.
For notational convenience we consider only the case $k-g+1=3$.
\begin{defi}	
	For $m \in \mathbb{N}$, let $U := \binom{[m]}{3} \cup \{\bot\}$.
	That is, we will index coordinates with triples $\{i_1,i_2,i_3\}$ in $[m]$,
	with one additional special coordinate $\bot$.
	For $i \in [m]$, let $S_i \subseteq U$ be the set of triples containing $i$.
	Let $\mathscr{S}_m$ be the family of all supersets of sets in $\{S_1,\dots,S_{m},\{\bot\}\}$,
	ordered so that sets not containing $\bot$ are all earlier than sets containing $\bot$.
	Let $q_m := q_{\mathscr{S}_m}$.
\end{defi}

Observe that the minimal smug sets of $q_m$ are exactly $S_1,\dots,S_{m},\{\bot\}$.
Note also that every two sets in $S_1,\dots,S_m$ (and hence any two of their supersets) intersect,
so the ordering between them is irrelevant, and similarly for every two sets containing $\bot$;
hence $q_m$ is defined unambiguously.

\begin{prop}\label{prop:qm}
	Let $m \geq 4$, $k-g+1=3$, and $\frac{g}{k} > \frac{1}{2}$.
	Then
	\begin{enumerate}[label={(\roman*)}]
		\item $q_m$ is a polymorphism of $(1,g,k)$-SetSAT of arity $\binom{m}{3}+1$;
    \item \label{qm:2} $q_m$ and projections of arity $m+1$ satisfy a $\frac{4}{m}$-robust minor condition;
		\item for every partial assignment to less than $\frac{m}{3}$ coordinates of $q_m$ and every value $a \in [s+1]$, there is an assignment to the remaining coordinates that makes $q_m$ take value $a$.
		(In particular this means $q_m$ does not have small ``weakly fixing'' or ``avoiding'' sets).
	\end{enumerate}
\end{prop}
\begin{proof}
To check (i), we have to check that $\mathscr{S}_m$ satisfies the two conditions of  Definition~\ref{def:smugToPolym}.
The first condition is trivial because all sets containing $\bot$ are in $\mathscr{S}_m$.
To check the second condition, suppose for contradiction that 
some $k$-tuple of sets in $\mathscr{S}_m$ covers every coordinate at most $k-g=2$ times.
In particular $\bot$ would be covered at most 2 times,
leaving at least $g \geq k-g+1 = 3$ sets not containing $\bot$.
By definition of $\mathscr{S}_m$ these three sets would be supersets of $S_{i_1},S_{i_2},S_{i_3}$ respectively for some $i_1,i_2,i_3 \in [m]$ (not necessarily distinct),
hence taking $I \in \binom{[m]}{3} \subseteq U$ to be any triple containing $\{i_1,i_2,i_3\}$,
we see that the coordinate $I$ is covered by all three sets, and hence by some $3 = k-g+1$ sets.

For (ii), we start with an informal description; the formal argument is below. 
Let us first consider a minor of $q_m$ defined by identifying all coordinates that are triples containing some $i \in [m]$.
Observe that this minor is a projection to the resulting coordinate, for all $i \in [m]$.
This gives $m$ identities between $q_m$ and a projection $p$.
However, the same identities could be satisfied by replacing $q_m$ with a projection to $\bot$;
to avoid this, we map $\bot$ to a different coordinate of $p$ for each $i \in [m]$.

Formally, let $p \colon [s+1]^{m+1} \to [s+1]$ be the projection of arity $m+1$ to the last coordinate,
$p(x_1,\dots,x_{m+1})=x_{m+1}$.
For $i \in [m]$, let $\pi_i \colon U \to [m+1]$ be defined as $\pi_i(I) = m+1$ if $I \in \binom{[m]}{3}$ and $I \ni i$, otherwise set $\pi_i(I) = i$ (in particular $\pi_i(\bot)=i$).
Then $q_m \xrightarrow{\pi_i} p$ for each $i \in [m]$.

Consider the bipartite minor condition $\Sigma$ with two symbols $f,g$ of arity $|U|$ and $m+1$, respectively,
and $m$ identities $f \xrightarrow{\pi_i} g$.
Clearly this condition is satisfied by $q_m,p$.
We claim the condition is $\frac{4}{m}$-robust, that is, no four identities out of the $m$ identities of $\Sigma$ can be simultaneously satisfied by projections.
Suppose the opposite, that is, assigning $f=p_{I}$ for some $I \in U$ and $g = p_i$ for some $i \in [m+1]$ satisfies four identities.
Without loss of generality these identities are $p_{I} \xrightarrow{\pi_1} p_i$, $p_{I} \xrightarrow{\pi_2} p_i$, $p_{I} \xrightarrow{\pi_3} p_i$, and $p_{I} \xrightarrow{\pi_4} p_i$.
Equivalently, $\pi_1(I) = i$, $\pi_2(I) = i$, $\pi_3(I) = i$, and $\pi_4(I) = i$.
The first condition implies that $i$ is either $1$ or $m+1$;
similarly the second implies that $i$ is either $2$ or $m+1$;
hence $i=m+1$.
The condition $\pi_1(I) = m+1$ then implies that $I$ is a triple in $\binom{[m]}{3}$ containing $1$.
Similarly $I$ must contain $2$, $3$, and $4$.
This is a contradiction, so $\Sigma$ is indeed $\frac{4}{m}$-robust.

For (iii), consider a partial assignment to some $k < \frac{m}{3}$ coordinates $I_1,\dots,I_k$ of $q_m$.
Let $I$ be the set of values $i \in [m]$ that are contained in some triple among $I_1,\dots,I_k$.
Then $|I| \leq 3k < m$, so there is a value $i^* \in [m]\setminus I$.
This means no coordinate in $S_{i^*}$ has been assigned yet.
Therefore, for any $a \in [s+1]$, assigning the value $a$ to all coordinates in $S_{i^*}$ (and remaining coordinates arbitrarily) makes $q_m$ take the value $a$.
\end{proof}

As a side note, another way to obtain a projection as a minor of $q_m$ is as follows.
Let $T \subseteq U$ be any set intersecting each of $S_1,\dots,S_m,\{\bot\}$.
Then identifying all coordinates in $T$ yields a projection to the resulting coordinate;
indeed, for any input $\bar{x} \in [s+1]^U$, the smug set of $\bar{x}$ in $q_m$ contains one of $S_1,\dots,S_m,\{\bot\}$ and hence contains a coordinate in $T$.

\begin{cor}\label{cor:34}
	Suppose $k-g+1=3$ and $g \geq 3$.
	Then the polymorphisms of $(1,g,k)$-SetSAT do not admit a minion homomorphism to a minion of bounded essential arity (or in fact to any minion with functions of arity $m$ having essential arity at most $\frac{m^{1/3}}{4}$).
\end{cor}

Therefore, by Corollary~\ref{cor:34}, the bounded essentially arity assumption
of Theorem~\ref{thm:minbndarity} does not apply and thus NP-hardness
cannot be derived from Theorem~\ref{thm:minbndarity}.
By Proposition~\ref{prop:qm}~\ref{qm:2}, Theorem~\ref{thm:robust} does not apply
either and neither does the weaker assumption where constants (bounding
essential arity or $1/\epsilon$ in assumptions involving $\epsilon$-robustness)
can be replaced by functions subpolynomial in the arity of the polymorphisms (as
in~\cite[Theorem~5.10]{bbko19}).

Even with the most general assumption used in~\cite[Theorem~5.22]{bbko19}, which
is also proved using a layered version of the PCP theorem, we were unable to
prove hardness (specifically in our Corollary~\ref{cor:smug}) despite several
attempts. On the other hand, we were unable to construct polymorphisms to show
that this general assumption fails for SetSAT.

We now turn to another sufficient condition for NP-hardness based on so-called Ol\v{s}\'{a}k functions.
Dinur, Regev, and Smyth~\cite{Dinur05:combinatorica} proved that the following PCSP is NP-hard, for any $k$: given a 3-uniform hypergraph that is 2-colourable, find a $k$-colouring.
One can hence deduce hardness of a PCSP by giving a minion homomorphism to the polymorphisms of this problem.
This was used by the authors of~\cite{bbko19} to improve the state-of-the-art for hardness of classical graph colouring approximation.
They also characterised when this approach is viable: such a minion homomorphism exists if and only if there is no \emph{Ol\v{s}\'{a}k function}, that is, a 6-ary function $o$ that satisfies
\begin{align*}
& o(x,x,y,y,y,x) \approx \\
\approx {} & o(x,y,x,y,x,y) \approx \\
\approx {} & o(y,x,x,x,y,y).
\end{align*}
(The six columns in this condition correspond to the satisfying assignments of
the problem of 2-colouring 3-uniform hypergraphs, or, equivalently, (monotone) Not-All-Equal 3-SAT).
We show that the polymorphisms of SetSAT include an Ol\v{s}\'{a}k function, proving that this approach is not viable for showing NP-hardness of SetSAT.

\begin{prop}\label{prop:olsak}
	Suppose $\frac{g}{k}>\frac{1}{2}$. There is a polymorphism of $(1,g,k)$-SetSAT with domain size $s+1$ that is an Ol\v{s}\'{a}k function.
\end{prop}
\begin{proof}
	Let us define three sets corresponding to positions of $x$ in the three rows defining an Ol\v{s}\'{a}k function: $S_1 = \{1,2,6\}, S_2=\{1,3,5\}, S_3=\{2,3,4\}$. Let $S_4$ be an arbitrary singleton, say $S_4=\{1\}$.
	Let $\mathscr{S}$ be the set of supersets of $S_1,S_2,S_3,S_4$,
	ordered so that supersets of $S_1,S_2,S_3$ come earlier.
	We claim the sequence of sets $\mathscr{S} = S_1,S_2,S_3,S_4$ satisfies the conditions of  Definition~\ref{def:smugToPolym}.
	The first condition is trivially satisfied because all sets containing $1$ are in $\mathscr{S}$.
	To check the second condition suppose for contradiction there is a $k$-tuple of sets in $\mathscr{S}$
	that covers every coordinate at most $k-g$ times.
	For each of these $k$ sets, choose one of the sets $S_1,S_2,S_3,S_4$ it contains.
	Let $n_1,n_2,n_3,n_4$ be the number of times we chose $S_1,S_2,S_3,S_4$, respectively.
	Then $n_1+n_2+n_3+n_4 = k$.
	Since the first coordinate (contained in $S_1,S_2,S_4$) is covered at most $k-g$ times, we have $n_1 + n_2 +n_4 \leq k-g$.
	Similarly, other coordinates give us inequalities $n_1+n_3 \leq k-g$, $n_2+n_3 \leq k-g$.
	This implies $k \leq n_1+2n_2+n_3+n_4 \leq 2(k-g)$ and hence $2g \leq k$.
	This contradicts $\frac{g}{k}>2$, so $\mathscr{S}$ satisfies the second condition of Definition~\ref{def:smugToPolym}.
	
	We claim that $q_{\mathscr{S}}$ is an Ol\v{s}\'{a}k function.
	Indeed, by definition, $q_{\mathscr{S}}(x,x,y,y,y,x) = x$, for all $x,y \in [s+1]$.
	Similarly $q_{\mathscr{S}}(x,y,x,y,x,y) = q_{\mathscr{S}}(y,x,x,x,y,y) = x$.
\end{proof}

Consider now a related condition. A \emph{Siggers function} is a 6-ary function $s$ that satisfies
\begin{align*}
s(x,y,x,z,y,z) \approx s(y,x,z,x,z,y).
\end{align*}
More generally, for a non-biparite graph $G$, a \emph{$G$-loop function} is a function satisfying the following condition,
where $(u_1,v_1),\dots,(u_{2m},v_{2m})$ lists both orientations of all $m$ edges of $G$:
\begin{align*}
  f(x_{u_1},x_{u_2},\dots, x_{u_{2m}}) \approx f(x_{v_1}, x_{v_2}, \dots, x_{v_{2m}}).
\end{align*}
The Siggers condition corresponds to $G = K_3$.
As shown in~\cite[Theorem 6.9]{bbko19}, the conjectured NP-hardness of the classical approximate colouring
is equivalent to the following statement:
for every $\A,\B$ such that $\Pol(\A,\B)$ contains no Siggers function, $\PCSP(\A,\B)$ is NP-hard.
In a similar way, $G$-loop functions characterise the conjectured NP-hardness of promise graph homomorphism (see~\cite[Theorem 6.12]{bbko19}).
In particular if $\Pol(\A,\B)$ contains no Siggers or $G$-loop function, for a non-bipartite $G$, this would imply that $\PCSP(\A,\B)$ is NP-hard conditional on those conjectures.

However, Siggers polymorphism and $G$-loop polymorphisms of $(1,g,k)$-SetSAT for $\frac{g}{k} > \frac{1}{2}$
are easily constructed similarly as in Proposition~\ref{prop:olsak}
(in the first case, it suffices to define $f(x_1,x_2,x_3,x_4,x_5,x_6)$ as $x_1$ if $x_1=x_3\ \,\wedge\,\ x_2=x_5\ \,\wedge\,\ x_3=x_6$ and $x_2$ otherwise).
Thus, even conditional NP-hardness of SetSAT would not follow this way.
As far as we know, SetSAT is the first known NP-hard promise CSP problem that admits $G$-loop polymorphisms.

\printbibliography

\appendix

\section{Simple reductions}
\label{app:simple}

\begin{prop}\label{prop:a1}
	For any $1\leq s < d$, there is a polynomial-time reduction from $(a,g,k)$-SetSAT to $(a,g,k+1)$-SetSAT. 
\end{prop}
\begin{proof}
	For $ i \in [s+1]$, let $N_i(x)=\mathbbm{1}[x \in [s+1]\setminus\{i\}]$ be a literal not satisfied by $i$. For each original clause, create $s+1$ new clauses by adding in turn each of the literals $N_i(y)$ to the original clause, where $y$ is a variable not appearing in the original instance. Note that a $g$-satisfying assignment to the original instance is also a $g$-satisfying assignment to the new instance. Conversely, if the original instance is not $a$-satisfiable, then neither is the new instance, as the literals $N_i(y)$ cannot simultaneously be satisfied in all the new clauses for any value of $y$.
\end{proof}

NP-hardness for the case $\frac{g}{k} \leq \frac{1}{2}$ is much easier to obtain
when $d \geq 3$ compared to $d=2$~\cite{agh17}, as shown below in
Corollary~\ref{padding}. First we prove more directly
that the generalisation of (1,1,3)-SAT to larger domains remains hard.

\begin{prop}
Let $s \geq 2$ and $d=s+1$. Then $(1,1,3)$-SetSAT is NP-hard.
\end{prop}

\begin{proof}
We give a reduction from 3-SAT. We interpret the value 1 as false and 2 as true. To illustrate the reduction, consider the clause
  $(x_1 \vee \overline{x}_2 \vee \overline{x}_3)$. From this clause we create a
  new clause $(N_1(x_1) \vee N_2(x_2) \vee N_2(x_3))$ where $N_i(x)$ is defined
  as in the proof of Proposition~\ref{prop:a1}. We add such a clause for each clause in the original 3-SAT instance. Then to enforce the binary nature of the original variables, we add the clauses $(N_i(x_j) \vee N_i(x_j) \vee N_i(x_j))$ for $3\leq i \leq s+1$ and $1\leq j \leq n$, which restrict the new variables to take values in $\{1,2\}$.

If the 3-SAT instance is satisfiable, then the (1,1,3)-SetSAT instance is also satisfiable. Conversely, if the SetSAT instance is satisfiable, then its variables take only the values 1 and 2 and we can translate back to a satisfying assignment for the 3-SAT instance.
\end{proof}

The first deviation from the results of the SAT world is that the SetSAT
analogue of 2-SAT is hard except in the case $s=1$, which corresponds to 2-SAT.
	
\begin{prop}
Let $s \geq 2$ and $d=s+1$. Then $(1,1,2)$-SetSAT is NP-hard.
\end{prop}

\begin{proof}
Each clause can be represented as $S_{a,b}(x,y)$, which forbids $(x,y)=(a,b)$. Therefore $\bigwedge_{i=1}^{s+1} S_{i,i}(x,y)$ is the disequality relation $x \neq y$, so we can simulate graph colouring, which is NP-hard for $s+1\geq 3$ colours.
\end{proof}

We can extend this result to larger clauses as follows.

\begin{cor}\label{padding}
Let $s \geq 2$ and $d=s+1$. Then $(1,g,2g)$-SetSAT is NP-hard for all $g \geq 1$.
\end{cor}

\begin{proof}
A reduction from $(1,1,2)$-SetSAT analogous to the reduction from $(1,1,3)$-SAT to $(1,g,3g)$-SAT in~\cite{agh17} gives the result. The clauses of the new $(1,g,2g)$-SetSAT instance are obtained by taking the union of all $g$-tuples of clauses from the $(1,1,2)$-SetSAT instance.

If the $(1,1,2)$-SetSAT instance is satisfiable, then the obtained $(1,g,2g)$-SetSAT instance is $g$-satisfiable for the same assignment. Conversely, since the $(1,g,2g)$-SetSAT instance contains clauses made from copying an old clause $g$ times, satisfiability of the new formula implies satisfiability of the old one, again for the same assignment.
\end{proof}

Finally we show that certain results on hypergraph colouring hardness
obtained by Guruswami and Lee~\cite{gl18} already imply NP-hardness fairly close to the real boundary.
\begin{prop}\label{colour}
For $d=s+1$ and all $g \geq 1$, $(1,s(g-1),(s+1)g)$-SetSAT is NP-hard. 
\end{prop}

\begin{proof}
We give a reduction from the following hypergraph colouring problem, whose hardness was proved in~\cite{gl18}. For $g,r,c\geq 2$, given as input a $gr$-uniform hypergraph that is promised to have an $r$-colouring where each colour appears at least $g-1$ times in every hyperedge, find a $c$-colouring that does not create a monochromatic hyperedge. The hardness reduction is as follows.

Let $r=c=s+1$. For each hyperedge $\{x_1,\ldots,x_{(s+1)g}\}$ we create, for $1 \leq i \leq s+1$, the SetSAT clauses $C_i = \left(N_i(x_1) \vee \ldots \vee N_i(x_{(s+1)g})\right)$, where $N_i$ are as in the proof of Proposition~\ref{prop:a1}. If the hypergraph instance has an $(s+1)$-colouring where every colour appears at least $g-1$ times in each hyperedge, then the obtained formula will be $s(g-1)$ satisfiable: under the promised assignment, the clause $C_i$ contains the group of satisfied literals whose variables are not equal to $i$, and there are at least $g-1$ literals in each of the $s$ such groups.

Conversely, if the SetSAT formula is satisfied, the variables $x_1,\ldots,x_{(s+1)g}$ cannot all take the same value $i$ for any $i$, as otherwise the clause $C_i$ would be false. Therefore no hyperedge in the hypergraph is left monochromatic by a satisfying assignment. 
\end{proof}

\end{document}

\section{Layered label cover: proof of Theorem~\ref{thm:layeredGLC}}
\label{app:layered}

In this section we adapt a reduction from the work of Dinur,
Guruswami, Khot, and Regev~\cite{DinurGKR05}.
Recall that an \emph{$\ell$-layered label cover} instance is a sequence of $\ell+1$ sets $X_0,\dots,X_\ell$ (called \emph{layers}) of variables with range $[m]$, for some \emph{domain size} $m\in\mathbb{N}$, and a set of constraints $\Phi$.
Each constraint is a function (often called a projection constraint) from a variable $x \in X_i$ to a variable in a further layer $y \in X_j$, $i < j$: that is, a function denoted $\phi_{x\to y}$ which is satisfied by an assignment $\sigma\colon X_0 \cup \dots \cup X_\ell \to [m]$ if $\sigma(y) = \phi_{x\to y}(\sigma(x))$.
A \emph{chain} is a sequence of variables $x_i \in X_i$ for $i=0,\dots,\ell$ such that there are constraints $\phi_{x_i \to x_j}$ between them, for $i < j$.
A chain is \emph{weakly satisfied} if at least one of these constraints is satisfied.

\begin{thm*}[Theorem~\ref{thm:layeredGLC} restated]
	For every $\ell \in \mathbb{N}$ and $\varepsilon>0$, there is an $m\in\mathbb{N}$ such that
	it is NP-hard to distinguish $\ell$-layered label cover instances with domain size $m$ that are fully satisfiable
	from those where not even an $\varepsilon$-fraction of all chains is weakly satisfied.
\end{thm*}

\begin{proof}
	For $\ell=1$ a chain consists of just one constraint,
	so weakly satisfying the chain is the same as satisfying its constraint.
	The claim is then equivalent to the hardness of the standard bipartite gap label cover problem,
	which holds even for \emph{bi-regular} instances: that is, instances $(Y,Z)$ such that
	every variable in $Y$ occurs in constraints with exactly $d_{+}$ variables in $Z$ and
	every variable in $Z$ occurs in constraints with exactly $d_{-}$ variables in $Y$,
	for some $d_{+},d_{-} \in \mathbb{N}$. (This hardness follows from the PCP
	theorem~\cite{Arora1998:jacm-proof,Arora98:jacm} and Raz's parallel repetition
  theorem~\cite{Raz98:sicomp}, cf.~\cite{DinurGKR05} and~\cite{Arora09:book}.)
	
	For $\ell>1$ we reduce from a bi-regular instance of bipartite gap label cover with variable sets $Y$ and $Z$, domain size $m$, constraints $\Gamma$ and gap $\varepsilon' := \varepsilon/\binom{\ell+1}{2}$.
	Let the domain size of the constructed instance $\Phi$ be $m^\ell$.
	Let the variable sets be $X_i := Z^i \times Y^{\ell-i}$ for $i=0,\dots,\ell$
	(that is, $\ell$-tuples of $i$ variables from $Z$ followed by $\ell-i$ variables from $Y$;
	this makes indices notationally more convenient than the other way around).
	Let the constraints between $X_i$ and $X_{j}$ (for $0\leq i<j \leq \ell$) be defined for pairs of tuples $\bar{x}$ and $\bar{x}'$ of the form:
	\begin{align*}
	\bar{x}=(z_1,\dots,z_i,\ \ &y_{i+1},\dots,y_j,\ \ y_{j+1},\dots,y_\ell)\in X_i
	\text{\quad and}\\
	\bar{x}'=(z_1,\dots,z_i,\ \ &z_{i+1},\dots,z_j,\ \ y_{j+1},\dots,y_\ell)\in X_j
	\end{align*}
	such that the original instance has a constraint $\phi_{y_{k}\to z_{k}} \in \Gamma$ for $k= i+1,\dots,j$.
	Let the new projection constraint $\phi_{\bar{x}\to\bar{x}'}$ map $(a_1,\dots,a_\ell)$ to $(b_1,\dots,b_\ell)$
	where $b_k := \phi_{y_k\to z_k}(a_k)$ for $k= i+1,\dots,j$ and $b_k := a_k$ otherwise.
	This concludes the construction.
	
	Note that chains in this instance are in bijection with $\ell$-tuples of original constraints in $\Gamma$.
	Indeed, a chain $\bar{x}_i \in X_i$ ($i=0,\dots,\ell$) is determined by $\bar{x}_0=(y_1,\dots,y_\ell)$ and $\bar{x}_\ell=(z_1,\dots,z_\ell)$
	such that $\Gamma$ has constraints $\phi_{y_k\to z_k}$ for $k=1,\dots,\ell$.
	Moreover, for each $i<j$, every constraint $\phi_{\bar{x}\to\bar{x}'}$ between $\bar{x}\in X_i$ and $\bar{x}'\in X_j$ appears in the same number of chains
	(namely $d_{-}^{\ \,i}\cdot d_{+}^{\ \ell-j}$).
	
	If the original instance $\Gamma$ was fully satisfiable then so is the new one $\Phi$:
	indeed, if $\sigma$ is a satisfying assignment for $\Gamma$,
	then $\bar{x} \mapsto(\sigma(x_1),\dots,\sigma(x_\ell))$ is a satisfying assignment for $\Phi$.
	
	Suppose now that in $\Phi$, an assignment $\sigma\colon X_0\cup\dots\cup X_\ell \to [m]^\ell$
	weakly satisfies at least $\varepsilon$ of all chains.
	Then there exists $0\leq i < j \leq \ell$ such that
	at least $\varepsilon/\binom{\ell+1}{2}=\varepsilon'$ of all chains
	are weakly satisfied at a constraint between $X_i$ and $X_j$.
	Every constraint between $X_i$ and $X_j$ is contained in the same number of chains, say $C$,
	hence at least $\varepsilon'$ of the constraints between $X_i$ and $X_j$ are satisfied
	(indeed, the number of thus satisfied chains is exactly $C$ times the number of satisfied constraints; similarly, the number of all chains is exactly $C$ times the number of all constraints between $X_i$ and $X_j$).
	
	Choose an arbitrary coordinate $k$ in $i+1,\dots,j$.
	Partition $X_i$ into equivalence classes such that $\bar{x},\bar{x}'$ are in the same class if they are identical on all coordinates except possibly coordinate $k$.
	Partition $X_j$ in the same way.
	There exists a pair of classes between which constraints exist and at least $\varepsilon'$ of them are satisfied.
	That is, there are
	\begin{align*}
	x_1,\dots,x_{k-1},\ \ &x_{k+1},\dots,x_\ell \in Y \cup Z\text{\quad and}\\
	x'_1,\dots,x'_{k-1},\ \ &x'_{k+1},\dots,x'_\ell \in Y \cup Z
	\end{align*}	
	such that $\sigma$ satisfies at least $\varepsilon'$ of the constraints between pairs of the form
	\begin{align*}
	(x_1,\dots,x_{k-1},\ y,\ x_{k+1},\dots,x_\ell) \in X_i \quad\ \quad\\
	(x'_1,\dots,x'_{k-1},\ z,\ x'_{k+1},\dots,x'_\ell) \in X_j \quad\ \quad
	\end{align*}
	where a constraint $\phi_{y\to z}$ exists in $\Gamma$.
	Therefore, one can define an assignment $\sigma'\colon Y\cup Z \to [m]$ 
	by letting $\sigma'(y)$ and $\sigma'(z)$ be the $k$-th element of the value in $[m]^\ell$ resulting from applying $\sigma$ to the above tuples, respectively for $y \in Y$ and $z \in Z$.
	This assignment then satisfies at least $\varepsilon'$ of all the constraints $\phi_{y\to z}$ of the original instance $\Gamma$.
\end{proof}